\documentclass[11pt]{article}
\usepackage{fullpage}
\usepackage{here}

\usepackage{amsthm,amsmath,amssymb}
\usepackage{xcolor}
\usepackage{graphicx}
\usepackage{booktabs}
\usepackage{caption}
\usepackage{hyperref}
\usepackage{enumerate}
\usepackage{enumitem}
\usepackage{lineno}

\theoremstyle{plain}
\newtheorem{theorem}{Theorem}
\newtheorem{lemma}[theorem]{Lemma}

\newtheorem{corollary}[theorem]{Corollary}

\newtheorem{observation}[theorem]{Observation}
\newtheorem{claim}[theorem]{Claim}
\newtheorem{remark}[theorem]{Remark}

\title{Minimum Weight Euclidean $(1+\eps)$-Spanners\thanks{A preliminary version of this paper appears in the
\emph{Proceedings of the 48th International Workshop on Graph-Theoretic Concepts in Computer Science (WG)}, LNCS, Springer, 2022.}
}

\author{Csaba D. T\'oth\thanks{Department of Mathematics, California State University Northridge, 18111 Nordhoff St., Los Angeles, CA~91330; Department of Computer Science, Tufts University, Medford, MA, USA.
Email: \texttt{csaba.toth@csun.edu}.
Research supported, in part, by the NSF awards DMS-1800734 and DMS-2154347.}}

\date{}

\newcommand{\MST}{{\rm MST}}

\newcommand{\N}{\mathbb{N}} %  set of natural numbers
\newcommand{\Z}{\mathbb{Z}} %  set of intgers
\newcommand{\R}{\mathbb{R}} %  set of real numbers

\newcommand{\eps}{\varepsilon}
\newcommand{\area}{\mathrm{area}}

\def\conv{\texttt{conv}}

\newcommand{\later}[1]{{}}
\newcommand{\old}[1]{{}}
\long\def\ignore#1{}

\begin{document}
\maketitle
%\linenumbers

\begin{abstract}
Given a set $S$ of $n$ points in the plane and a parameter $\varepsilon>0$, a Euclidean $(1+\varepsilon)$-spanner is a geometric graph $G=(S,E)$ that contains, for all $p,q\in S$, a $pq$-path of weight at most $(1+\varepsilon)\|pq\|$.
We show that the minimum weight of a Euclidean $(1+\varepsilon)$-spanner for $n$ points in the unit square $[0,1]^2$ is $O(\varepsilon^{-3/2}\,\sqrt{n})$, and this bound is the best possible.
The upper bound is based on a new spanner algorithm in the plane.
It improves upon the baseline $O(\varepsilon^{-2}\sqrt{n})$, obtained by combining a tight bound for the weight of a Euclidean minimum spanning tree (MST) on $n$ points in $[0,1]^2$, and a tight bound for the lightness of Euclidean $(1+\varepsilon)$-spanners, which is the ratio of the spanner weight to the weight of the MST.
Our result generalizes to Euclidean $d$-space for every constant dimension $d\in \mathbb{N}$: The minimum weight of a Euclidean $(1+\varepsilon)$-spanner for $n$ points in the unit cube $[0,1]^d$ is $O_d(\varepsilon^{(1-d^2)/d}n^{(d-1)/d})$, and this bound is the best possible.

For the $n\times n$ section of the integer lattice in the plane, we show that the minimum weight of a Euclidean $(1+\varepsilon)$-spanner is between $\Omega(\varepsilon^{-3/4}\cdot n^2)$ and $O(\varepsilon^{-1}\log(\varepsilon^{-1})\cdot n^2)$. These bounds become $\Omega(\varepsilon^{-3/4}\cdot \sqrt{n})$ and $O(\varepsilon^{-1}\log(\varepsilon^{-1})\cdot \sqrt{n})$ when scaled to a grid of $n$ points in the unit square. In particular, this shows that the integer grid is \emph{not} an extremal configuration for minimum weight Euclidean $(1+\varepsilon)$-spanners.\\

\end{abstract}

\section{Introduction}  \label{sec:intro}

For a set $S$ of $n$ points in a metric space and a parameter $t\geq 1$, a graph $G=(S,E)$ is a \emph{$t$-spanner} if $G$ contains, between every two points $p,q\in S$, a $pq$-path of weight at most $t\cdot \|pq\|$, where $\|pq\|$ denotes the distance between $p$ and $q$. In other words, a $t$-spanner distorts the true distances between any pair of points up to a factor $t$; the parameter $t$ is the \emph{stretch factor} of the spanner. Several optimization criteria have been considered for $t$-spanners, such as the \emph{size} (number of edges), the \emph{weight} (sum of edge weights), the \emph{maximum degree}, and the \emph{hop-diameter}. Specifically, the \emph{sparsity} of a spanner is the ratio $|E|/|S|$ between the number of edges and vertices; and the \emph{lightness} is the ratio between the weight of a spanner and the weight of a minimum spanning tree (MST) on $S$.

In this paper, we focus on the Euclidean $d$-space in constant dimension $d\in \N$. In this case, for every constant $\eps>0$, there exist $(1+\eps)$-spanners with sparsity and lightness that depends only on $\eps$ and $d$.
The precise dependence on $\eps$ (assuming constant dimension $d$) has almost completely been determined:
Every finite point set in $\R^d$ admits a Euclidean $(1+\eps)$-spanner with $O_d(\eps^{1-d})$ sparsity and $O_d(\eps^{-d} \log\eps^{-1})$ lightness. The sparsity bound is the best possible, and the lightness bound is tight up to the $\log\eps^{-1}$ factor~\cite{LeS22}. Several classical spanner constructions achieve sparsity $O_d(\eps^{1-d})$ in $\R^d$ such as Yao-graphs~\cite{althofer1993sparse,Clarkson87,ruppert1991approximating,Yao82}, $\Theta$-graphs~\cite{Kei88,KeilG92}, gap-greedy spanners~\cite{AryaS97} and path-greedy spanners~\cite{althofer1993sparse}; see the book by Narasimhan and Smid~\cite{NS-book} and the more recent surveys
surveys~\cite{GudmundssonK18,MitchellW18}.
For lightness, Das et al.~\cite{DasHN93,DasNS95,NS-book} were the first to construct $(1+\eps)$-spanners of lightness $\eps^{-O(d)}$ in Euclidean $d$-space. This bound has been generalized to metric spaces with doubling dimension $d$, see~\cite{BorradaileLW19,Gottlieb15,FiltserS20}. Recently, Le and Solomon~\cite{LeS22} showed that the path-greedy $(1+\eps)$-spanner in $\R^d$ has lightness $O(\eps^{-d}\log \eps^{-1})$; and so the same spanner achieves a tight bound for sparsity and an almost tight bound for lightness.

\paragraph{Lightness versus Minimum Weight.}
Lightness is a convenient optimization parameter for a spanner on a finite point set $S\subset \R^d$, as it is invariant under similarities (i.e., dilations, rotations, and translations). It also provides an approximation ratio for the minimum weight $(1+\eps)$-spanner on $S$, as the weight of a Euclidean MST (for short, EMST) of $S$ is a trivial lower bound on the spanner weight.
For $n$ points in the unit cube $[0,1]^d$, the weight of the EMST is $O_d(n^{1-1/d})$, and this bound is the best possible~\cite{Fe55,SteeleS89}. In particular, a suitably scaled section of the integer lattice attains this bound up to constant factors (in every constant dimension $d\in \N$). Supowit et al.~\cite{SupowitRP83} proved similar bounds for the minimum weight of other popular graphs, such as spanning cycles and perfect matchings on $n$ points in the unit cube $[0,1]^d$.

The combination of the current best bounds on lightness and the weight of an EMST implies that every set of $n$ points in $[0,1]^d$ admits a Euclidean $(1+\eps)$-spanner of weight
\begin{equation}\label{eq:naive}
 O\left(\eps^{-d} \log \eps^{-1} \cdot n^{1-1/d}\right).
\end{equation}
However, the lower bound constructions for the two (almost) tight upper bounds are very different:
The $\Omega_d(\eps^{-d})$ lightness lower bound holds for the union of $(d-1)$-dimensional grids on two opposite facets of a unit cube~\cite{LeS22}, and the $\Omega_d(n^{1-1/d})$ bound on EMST holds for a $d$-dimensional grid in a unit cube~\cite{Fe55}.
It turns out that these upper bounds cannot be attained simultaneously.
In fact, the bound~\eqref{eq:naive} can be improved to
\begin{equation}\label{eq:new}
    O_d\left(\eps^{(1-d^2)/d}\cdot n^{1-1/d}\right)
\end{equation}
for every constant $d\geq 2$, and \eqref{eq:new} is tight (Theorem~\ref{thm:cube}).
The improvement is most prominent in the plane:
from $O(\eps^{-2}\sqrt{n})$ to $O(\eps^{-3/2}\sqrt{n})$.

\paragraph{Contributions.}
We obtain a tight upper bound on the minimum weight of a Euclidean $(1+\eps)$-spanner for $n$ points in $[0,1]^d$.
\begin{theorem}\label{thm:cube}
For every constant $d\in \N$ and $\eps>0$, every set of $n$ points in the unit cube $[0,1]^d$ admits a Euclidean $(1+\eps)$-spanner
of weight $O_d(\eps^{(1-d^2)/d}n^{(d-1)/d})$, and this bound is the best possible.
\end{theorem}
The upper bound is established by a new spanner algorithm, \textsc{SparseYao}, that modifies the classical Yao-graph construction using novel geometric insight (Section~\ref{sec:alg}). The weight analysis is based on a charging scheme that charges the weight of the spanner to empty regions (Section~\ref{sec:square}).
The lower bound construction is a lattice generated by orthogonal vectors of weight $\sqrt{\eps}$ and $\frac{1}{\sqrt{\eps}}$ (Section~\ref{sec:lower}); rather than a square grid.

A section of the integer lattice is the extremal configuration for many problems in combinatorial geometry (e.g., minimum-weight EMST, as noted above). Surprisingly, the minimum weight of Euclidean $(1+\eps)$-spanners for such a grid is still not known. We provide the first nontrivial upper and lower bounds in the plane.
\begin{theorem}\label{thm:grid}
For every $n\in \N$, the minimum weight of a $(1+\eps)$-spanner for the $n\times n$ section of the integer lattice is between $\Omega(\eps^{-3/4}\cdot n^2)$ and $O(\eps^{-1}\log\eps^{-1}\cdot n^2)$.
\end{theorem}

When scaled to $n$ points in the unit square $[0,1]^2$, the upper bound confirms that a square grid does not maximize the minimum weight of Euclidean $(1+\eps)$-spanners.
\begin{corollary}\label{cor:grid}
For every $n\in \N$, the minimum weight of a $(1+\eps)$-spanner for $n$ points in a scaled section of the integer grid in $[0,1]^2$ is
between $\Omega(\eps^{-3/4}\cdot \sqrt{n})$ and $O(\eps^{-1}\log\eps^{-1}\cdot \sqrt{n})$.
\end{corollary}

The lower bound is derived from an elementary criterion (the empty slab condition) for an edge to be present in every $(1+\eps)$-spanner (Section~\ref{sec:lower}). The upper bound is based on analyzing the \textsc{SparseYao} algorithm from Section~\ref{sec:alg}, combined with number theoretic results on Farey sequences (Section~\ref{sec:grid}). Closing the gap between the lower and upper bounds in Theorem~\ref{thm:grid} remains an open problem. Higher dimensional generalizations are also left for future work. In particular, multidimensional variants of Farey sequences are currently not well understood.

\paragraph{Further Related Previous Work.}
Many algorithms have been developed for constructing Euclidean $(1+\eps)$-spanners for $n$ points in $\R^d$~\cite{ABUAFFASH2022101807,ChanHJ20,DasHN93,DasN97,DasNS95,ElkinS15,Gottlieb15,GudmundssonLN02,LeS23,LevcopoulosNS02,RaoS98}, each designed for one or more optimization criteria (such as lightness, sparsity, hop diameter, maximum degree, or running time). See~\cite{GudmundssonK18,MitchellW18,NS-book} for comprehensive surveys. We briefly review previous constructions pertaining to the \emph{minimum weight} for $n$ points in the unit cube $[0,1]^d$,
where the weight of an EMST is $O(n^{1-1/d})$.

The path-greedy spanner algorithm~\cite{althofer1993sparse} (abbreviated as ``greedy'') generalizes Kruskal's MST algorithm: Given $\eps>0$ and $n$ points in a metric space, it sorts the ${n\choose 2}$ edges of the complete graph $K_n$ on the $n$ points by nondecreasing weight, and incrementally constructs a spanner $H$ starting from the empty graph; it adds an edge $uv$ to the spanner if $H$ does not already contain an $uv$-path of weight at most $(1+\eps)\|uv\|$.
As noted above, Le and Solomon~\cite{LeS22} recently proved that for $n$ points in the unit cube, the greedy algorithm returns a $(1+\eps)$-spanner of weight $O(\eps^{-d}\log \eps^{-1}\cdot  n^{1-1/d})$. This bound is weaker than Theorem~\ref{thm:cube} for $n$ points arranged in a grid, for example, but it is unclear whether the analysis of the greedy spanner can be improved to match the tight bound in Theorem~\ref{thm:cube}.

A classical method for constructing a $(1+\eps)$-spanners uses \emph{well-separated pair decompositions} (\emph{WSPD}) with a hierarchical clustering; see~\cite[Chap.~3]{Sariel}. Due to a hierarchy of depth $\Theta(\log n)$, this technique has been adapted broadly to dynamic, kinetic, and reliable spanners~\cite{BuchinHO20,ChanHJ20,gao2006deformable,Roditty12}. However, the weight of the resulting $(1+\eps)$-spanner for $n$ points in $[0,1]^d$ is $O(\eps^{-(d+1)}\cdot n^{1-1/d}\cdot\log n)$~\cite{gao2006deformable}. The $O(\log n)$ factor is due to the depth of the hierarchy; and it cannot be removed for any spanner with hop-diameter $O(\log n)$~\cite{AgarwalWY05,DinitzES10,SolomonE14}.

Several constructions for $(1+\eps)$-spanners combine hierarchical clustering with a greedy approach~\cite{DasN97,GudmundssonLN02} (either using approximate distances or running a greedy algorithm on a hierarchical spanner). These ``approximate'' greedy spanners achieve lightness $\eps^{-O(d)}$  in Euclidean $d$-space, hence weight $\eps^{-O(d)}\cdot n^{1-1/d}$ for $n$ points in $[0,1]^d$. They have been generalized to doubling metrics with the same lightness bound~\cite{BorradaileLW19,FiltserS20,Gottlieb15}.

Yao-graphs and $\Theta$-graphs are geometric proximity graphs, defined in the plane as follows. For a constant $k\geq 3$, consider $k$ cones of aperture $2\pi/k$ around each point $p\in S$; in each cone, connect $p$ to a ``closest'' point $q\in S$. The two graphs (Yao-graphs and $\Theta$-graphs) differ in the metric used for measuring the distance from $p$ to $q$: For Yao-graphs, $q$ minimizes the \emph{Euclidean distance} $\|pq\|$ from $p$ to $p$, and for $\Theta$-graphs $q$ minimizes the distance from $p$ to the orthogonal projection of $q$ to the angle bisector of the cone.
Both constructions generalize to $\R^d$, for all $d\geq 2$, where the number of cones is $O_d(k^{1-d})$.
For $\Theta$-graphs, the computation of closest points reduces to orthogonal range searching queries (rather than circular disk ranges for Yao-graphs), and a $\Theta$-graph for $n$ points in $\R^d$ can be computed in $O_d(nk^{1-d}\log^{d-1} n)$ time~\cite[Sec.~5.5]{NS-book}.
Both $\Theta$- and Yao-graphs are $(1+\eps)$-spanners if $k\in \Theta(\eps^{-1})$, and this bound cannot be improved~\cite{althofer1993sparse,KeilG92,NS-book,ruppert1991approximating}. Finding the optimal constant factors hidden in the notation $\Theta(\eps^{-1})$ is an active area of research~\cite{AichholzerBBBKRTV14,AkitayaBB22,BarbaBDFKORTVX15,BoseCMRV16,BoseHO21}. However, if we place $\lfloor n/2\rfloor$ and $\lceil n/2\rceil$ equally spaced points on opposite sides of the unit square, then the weight of both Yao- and $\Theta$-graphs with parameter $k=\Theta(\eps^{-1})$ will be $\Theta(\eps^{-1}\, n)$, which substantially exceeds the bound in Theorem~\ref{thm:cube} for any fixed $\eps>0$ as $n\rightarrow \infty$.

Optimizing various parameters of Euclidean spanners on the integer lattice is a natural problem. Previous work determined the minimum stretch factor in terms of the maximum degree $\Delta$ for small values of $\Delta$. There exist Euclidean $(1+\sqrt{2})$-spanners for $\Z^2$ with maximum degree $\Delta=3$, and $\sqrt{2}$-spanners with $\Delta=4$; and these bounds are the best possible~\cite{DumitrescuG16,GalantP22}.

\paragraph{Organization.}
We start with our lower bound constructions (Section~\ref{sec:lower}) as a warm-up. The two elementary geometric criteria for inclusion in a Euclidean $(1+\eps)$-spanner build intuition and highlight the significance of $\sqrt{\eps}$ as the ratio between the two axes of an ellipse of all paths of stretch at most $1+\eps$ between the foci. Section~\ref{sec:alg} presents Algorithm \textsc{SparseYao} and its stretch analysis in the plane. Its weight analysis for $n$ points in $[0,1]^2$ is in Section~\ref{sec:square}. The algorithm and its analysis generalize fairly easily to Euclidean $d$-space for every constant $d\geq 2$; this is sketched in Section~\ref{sec:d-space}. We analyze the performance of Algorithm \textsc{SparseYao} on the $n\times n$ grid, after a brief review of Feray sequences, in Section~\ref{sec:grid}.
We conclude with a selection of open problems in Section~\ref{sec:con}.

\section{Lower Bounds in the Plane}
\label{sec:lower}

We present lower bounds for the minimum weight of a $(1+\eps)$-spanner for the $n\times n$ section of the integer lattice (Section~\ref{ssec:gridLB}); and for $n$ points in a unit square $[0,1]^2$ (Section~\ref{ssec:squareLB}).

Let $S\subset \R^2$ be finite. Every pair of points, $a,b\in S$, determines a (closed) line segment $ab=\conv \{a,b\}$; the relative interior of $ab$ is denoted by $\mathrm{int}(ab)=ab\setminus \{a,b\}$.
Let $\mathcal{E}_{ab}$ denote the ellipse with foci $a$ and $b$, and great axis of weight $(1+\eps)\|ab\|$, and let $\mathcal{L}_{ab}$ be the slab bounded by two lines parallel to $ab$ and tangent to $\mathcal{E}_{ab}$; see Fig.~\ref{fig:ellipse}.
Note that the width of $\mathcal{L}_{ab}$ equals the minor axis of $\mathcal{E}_{ab}$,
which is $((1+\eps)^2-1^2)^{1/2}\|ab\|=(2\eps+\eps^2)^{1/2}\|ab\|>\sqrt{2\eps}\|ab\|$.
When $S\subset \Z^2$, the vector $\overrightarrow{ab}$ is called \emph{primitive} if $\overrightarrow{ab}=(x,y)$ and $\mathrm{gcd}(x,y)=1$.

\begin{figure}[htbp]
\begin{center}
  \includegraphics[width=0.8\textwidth]{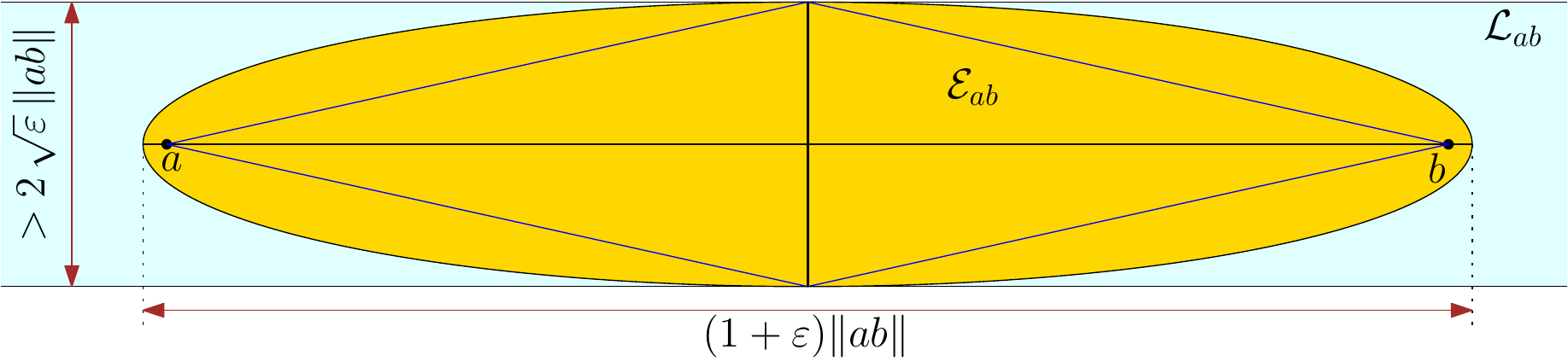}
\end{center}
\caption{The ellipse $\mathcal{E}_{ab}$ with foci $a$ and $b$, and great axis of weight $(1+\eps)\|ab\|$.}
\label{fig:ellipse}
\end{figure}

Consider the following two conditions for edges of a Euclidean $(1+\eps)$-spanner for $S$.
\begin{itemize}%\itemsep 0pt
\item \textbf{Empty ellipse condition}:  $S\cap \mathcal{E}_{ab}=\{a,b\}$.
\item \textbf{Empty slab condition}:  $S\cap \mathrm{int}(ab)=\emptyset$ and all points in $S\cap \mathcal{L}_{ab}$ are on the line $ab$.
\end{itemize}

\begin{observation}\label{obs:elementary}
Suppose $S\subset \R^2$ is finite, $G=(S,E)$ is a $(1+\eps)$-spanner for $S$, and $a,b\in S$.
\begin{enumerate}
\item If $ab$ meets the empty ellipse condition, then $ab\in E$.
\item If $S$ is a section of $\Z^2$, $\eps<1$, and $ab$ meets the empty slab condition, then $ab\in E$.
\end{enumerate}
\end{observation}
\begin{proof}
The ellipse $\mathcal{E}_{ab}$ contains all points $p\in \R^2$ satisfying $\|ap\|+\|pb\|\leq (1+\eps)\|ab\|$. By the triangle inequality, every point along an $ab$-path of weight at most $(1+\eps)\|ab\|$ satisfies this inequality, hence the entire path lies in $\mathcal{E}_{ab}$. The empty ellipse condition implies that such a path cannot have interior vertices.

If $S$ is the integer lattice, then $S\cap \mathrm{int}(ab)=\emptyset$ implies that $\overrightarrow{ab}$ is a primitive vector, hence the distance between any two lattice points along the line $ab$ is at least $\|ab\|$. Given that $\mathcal{E}_{ab}\subset \mathcal{L}_{ab}$, the empty slab condition now implies the empty ellipse condition.
\end{proof}

\subsection{Lower Bounds for the Grid}
\label{ssec:gridLB}

We can now establish the lower bound in Theorem~\ref{thm:grid}.

\begin{lemma}\label{lem:LBgrid}
For every $\eps\in (0,1]$ and every integer $n\geq 2\,\eps^{-1/4}$, the weight of every $(1+\eps)$-spanner for the $n\times n$ section of the integer lattice is $\Omega(\eps^{-3/4}n^2)$.
\end{lemma}
\begin{proof}
Let $S=\{(s_1,s_2)\in \Z^2: 0\leq s_1,s_2<n\}$ and $A=\{(a_1,a_2)\in \Z^2: 0\leq a_1,a_2<\lceil \eps^{-1/4}\rceil/2\}$.
Denote the origin by $o=(0,0)$. For every grid point $a\in A$, we have $\|oa\|\leq \eps^{-1/4}/\sqrt{2}$.
We show that every primitive vector $\overrightarrow{oa}$ with $a\in A$ satisfies the empty slab condition. It is clear that $S\cap \mathrm{int}(oa)=\emptyset$. Let $s\in S$ be a point that is not on the line $\mathrm{span}(oa)$ spanned by $oa$. By Pick's theorem, $\area(\Delta(oas))\geq \frac12$. Since $\area(\Delta(oas)) =\frac12\cdot \|oa\|\cdot \mathrm{dist}(\mathrm{span}(oa),s)$, then the distance between the line spanned by $oa$ and $s$ is at least $\|oa\|^{-1}\geq \sqrt{2}\cdot \eps^{1/4} \geq 2\, \eps^{1/2}\, \|oa\|$;
and so $s\notin \mathcal{L}_{oa}$, as claimed.

By elementary number theory, $\overrightarrow{oa}$ is primitive for $\Theta(|A|)$ points $a\in A$. Indeed, every $a_1\in \N$ is relatively prime to $N\varphi(a_1)/a_1$ integers in every interval of length $N$, where $\varphi(.)$ is Euler's totient function, and $\varphi(a_1)=\Theta(a_1)$. Consequently, the total weight of primitive vectors $\overrightarrow{oa}$, $a\in A$, is $\Theta(|A|\cdot \eps^{-1/4})=\Theta(\eps^{-3/4})$.

The edges $oa$, for all primitive vectors $\overrightarrow{oa}$ with $a\in A$, form a star centered at the origin. The translates of this star to other points $s\in S$, with $0\leq s_1,s_2\leq \frac{n}{2}\leq n - \lceil \eps^{-1/4}\rceil$ are also present in every $(1+\eps)$-spanner for $S$. As every edge is part of at most two such stars, summation over $\Theta(n^2)$ stars yields a lower bound of $\Omega(\eps^{-3/4}n^2)$.
\end{proof}

\begin{remark}
{\rm
The lower bound in Lemma~\ref{lem:LBgrid} derives from the total weight of primitive vectors $\overrightarrow{oa}$ with $\|oa\|\leq O(\eps^{-1/4})$, which satisfy the empty slab condition. There are additional primitive vectors that satisfy the empty ellipse condition (e.g., $\overrightarrow{oa}$ with $a=(1,y)$ for all $|y|<\eps^{-1/3}$). However, it is unclear how to account for all vectors satisfying the empty ellipse condition, and whether their overall weight would improve the lower bound in Lemma~\ref{lem:LBgrid}.
}
\end{remark}

\begin{remark}
{\rm
The empty ellipse and empty slab conditions each imply that an edge \emph{must} be present in every $(1+\eps)$-spanner for $S$. It is unclear how the total weight of such ``must have'' edges compare to the minimum weight of a $(1+\eps)$-spanner.
}
\end{remark}

\begin{remark}
{\rm
We obtain a grid of $\Theta(n)$ points in the unit square $[0,1]^2$ if we scale down the $\lceil \sqrt{n}\rceil\times \lceil \sqrt{n}\rceil$ section of the integer lattice by a factor of $\lceil \sqrt{n}\rceil$. The EMST of this grid has weight $\Theta(\sqrt{n})$. Le and Solomon~\cite{LeS22} proved that the greedy $(1+\eps)$-spanner has weight $O(\eps^{-2}\log \eps^{-1} \sqrt{n})$, although the true performance of the greedy algorithm might be better. In contrast, Theorem~\ref{thm:cube} yields a $(1+\eps)$-spanner of weight $O(\eps^{-3/2}\sqrt{n})$; and
Corollary~\ref{cor:grid} further improves it to $O(\eps^{-1}\log(\eps^{-1})\sqrt{n})$.
}
\end{remark}

\subsection{Lower Bounds in the Unit Square}
\label{ssec:squareLB}

We continue with the lower bound in Theorem~\ref{thm:cube}.

\begin{lemma}\label{lem:squareLB}
For every $\eps\in (0,1]$ and every integer $n\geq 2\,\lfloor \eps^{-1}/2\rfloor$, there exists a set $S$ of $n$ points in $[0,1]^2$ such that every $(1+\eps)$-spanner for $S$ has weight $\Omega(\eps^{-3/2}\sqrt{n})$.
\end{lemma}
\begin{proof}
First let $S_0$ be a set of $2m$ points, where $m=\lfloor \eps^{-1}/2\rfloor$, with $m$ equally spaced points on two opposite sides of a unit square. By the empty ellipse property, every $(1+\eps)$-spanner for $S_0$ contains a complete bipartite graph $K_{m,m}$. The weight of each edge of $K_{m,m}$ is between $1$ and $\sqrt{2}$, and so the weight of every $(1+\eps)$-spanner for $S_0$ is $\Omega(\eps^{-2})$.

For $n>2m$, consider an $\lfloor \sqrt{\eps n}\rfloor\times \lfloor \sqrt{\eps n}\rfloor$ grid of unit squares, and insert a translated copy of $S_0$ in each unit square. Let $S$ be the union of these $\Theta(\eps n)$ copies of $S_0$; and note that $|S|=\Theta(n)$. A $(1+\eps)$-spanner for each copy of $S_0$ still requires a complete bipartite graph of weight $\Omega(\eps^{-2})$. Overall, the weight of every $(1+\eps)$-spanner for $S$ is $\Omega(\eps^{-1}n)$.

Finally, scale $S$ down by a factor of $\lfloor \sqrt{\eps n}\rfloor$ so that it fits in a unit square. The weight of every edge scales by the same factor, and the weight of a $(1+\eps)$-spanner for the resulting $n$ points in $[0,1]^2$ is $\Omega(\eps^{-3/2}\, \sqrt{n})$, as claimed.
\end{proof}

\begin{remark}
{\rm
The points in the lower bound construction above lie on $O(\sqrt{\eps n})$ axis-parallel lines in $[0,1]^2$, and so the weight of their EMST is $O(\sqrt{\eps n})$. Recall that the lightness of the greedy $(1+\eps)$-spanner is $O(\eps^{-d}\log \eps^{-1})$~\cite{LeS22}. For $d=2$, it yields a $(1+\eps)$-spanner of weight $O(\eps^{-2}\log \eps^{-1})\cdot \|\MST(S)\|=O(\eps^{-3/2}\log (\eps^{-1}) \sqrt{n} )$.
Note that the greedy algorithm returns a $(1+\eps)$-spanner of nearly minimum weight in this case.
}
\end{remark}

\section{Spanner Algorithm: Sparse Yao-Graphs}
\label{sec:alg}

For clarity, we present our new spanner algorithm in full detail in the plane (in this section), and then sketch the straightforward generalization to $\R^d$ (in Section~\ref{sec:d-space}).
Let $S$ be a set of $n$ points in $\R^2$ and $\eps\in (0,\frac19)$. As noted above, the Yao-graph $Y_k(S)$ with $k=\Theta(\eps^{-1})$ cones per vertex is a $(1+\eps)$-spanner for $S$~\cite{althofer1993sparse,Clarkson87}.
We describe a new algorithm, \textsc{SparseYao}$(S,\eps)$, by modifying the classical Yao-graph construction (Section~\ref{ssec:alg}); and show that it returns a $(1+\eps)$-spanner for $S$ (Section~\ref{ssec:stretch}). Later, we use this algorithm for $n$ points in the unit square (Section~\ref{sec:square}; and for an $n\times n$ section of the integer lattice (Section~\ref{sec:grid}).

The basic idea is that, in some cases, cones of aperture $\Theta(\sqrt{\eps})$ suffice for the inductive proof that establishes the spanning ratio $1+\eps$. This idea is fleshed out in Section~\ref{ssec:stretch}.
The use of the angle $\sqrt{\eps}$ then allows us to charge the weight of the resulting spanner to the area of empty regions (specifically, to an empty section of a cone) in Section~\ref{sec:square}.

\subsection{Sparse Yao-Graph Algorithm}
\label{ssec:alg}

We present an algorithm that computes a modified version of the Yao-graph for a point set $S$ and parameter $\eps>0$. It starts with cones of aperture $\Theta(\sqrt{\eps})$, and refines them to cones of aperture $\Theta(\eps^{-1})$. We connect each point $p\in S$ to the closest point in each large cone, and add edges in the small cones only if necessary. To specify when exactly the small cones are needed, we define two geometric regions that will also play crucial roles in the stretch and weight analyses.

\begin{figure}[htbp]
\begin{center}
  \includegraphics[width=0.7\textwidth]{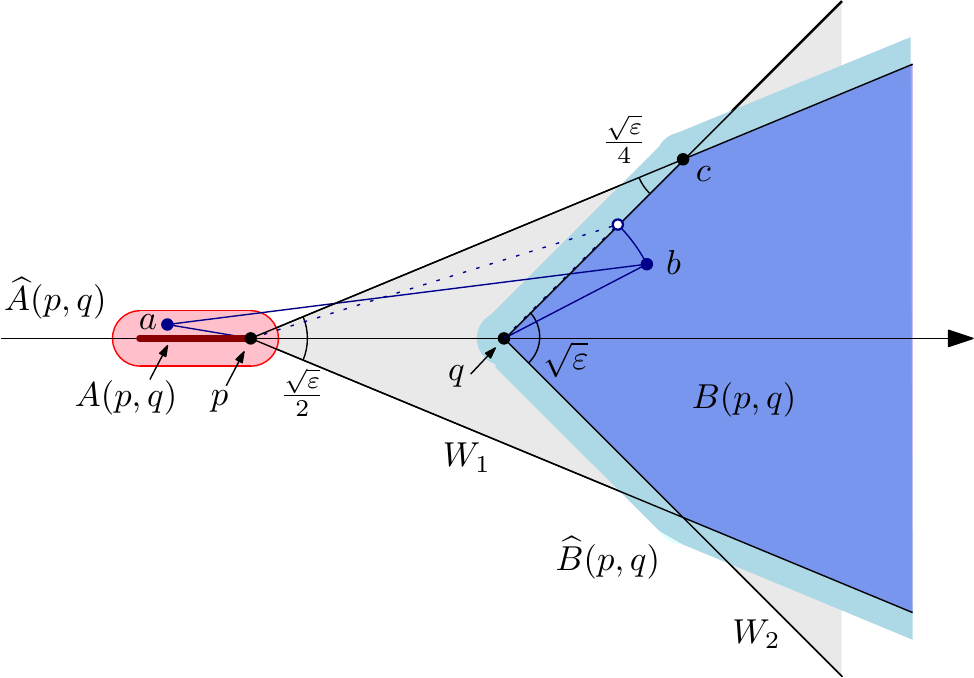}
\end{center}
\caption{Cones $W_1$ and $W_2$, line segment $A(p,q)$,
and regions $\widehat{A}(p,q)$, $B(p,q)$, and $\widehat{B}(p,q)$ for a pair of points $p,q\in S$.}
\label{fig:wedge}
\end{figure}

\smallskip\noindent\textbf{Definitions.}
Let $p, q\in S$ be distinct points; refer to Fig.~\ref{fig:wedge}.
Let $A(p,q)$ be the line segment of weight $\frac{\sqrt{\eps}}{16}\,\|pq\|$ on the line $pq$ with one endpoint at $p$ but interior-disjoint from the ray $\overrightarrow{pq}$; and $\widehat{A}(p,q)$ the set of points in $\R^2$ within distance $\frac{\eps}{64}\,\|pq\|$ from $A(p,q)$.
Let $W_1$ be the cone with apex $p$, aperture $\frac12\cdot \sqrt{\eps}$, and symmetry axis $\overrightarrow{pq}$; and let $W_2$ be the cone with apex $q$, aperture $\sqrt{\eps}$, and the same symmetry axis $\overrightarrow{pq}$. Let $B(p,q)=W_1\cap W_2$.
Finally, let $\widehat{B}(p,q)$ be the set of points in $\R^2$ within distance at most $\frac{\eps}{32}\,\|pq\|$ from $B(p,q)$.

We show below (Lemma~\ref{lem:technical3}) that if we add edge $pq$ to the spanner, then we do not need any of the edges $ab$ with $a\in \widehat{A}(p,q)$ and $b\in \widehat{B}(p,q)$. We can now present our algorithm.

\paragraph{Algorithm \textsc{SparseYao}$(S,\eps)$.}
Input: a set $S\subset \R^2$ of $n$ points, and $\eps\in (0,\frac19)$.

\smallskip\noindent\textbf{Preprocessing Phase.}
Subdivide $\R^2$ into $k:=\lceil 16\,\pi/\sqrt{\eps}\rceil$ congruent cones of aperture $2\pi/k\leq \frac18\cdot \sqrt{\eps}$ with apex at the origin, denoted $C_1,\ldots ,C_k$. For $i\in \{1,\ldots ,k\}$, let $\overrightarrow{r}_i$ be the symmetry axis of $C_i$ directed from the origin towards th interior of $C_i$.
For each $i\in \{1,\ldots , k\}$, subdivide $C_i$ into $k$ congruent cones of aperture $2\pi/k^2\leq \eps/8$, denoted $C_{i,1},\ldots , C_{i,k}$; see Fig.~\ref{fig:Yao}. For each point $s\in S$, let $C_i(s)$ and $C_{i,j}(s)$, resp., be the translates of cones $C_i$ and $C_{i,j}$ to apex $s$.

For all $s\in S$ and $i\in \{1,\ldots ,k\}$, let $q_i(s)$ be a closest point to $s$ in $C_i(s)\cap (S\setminus \{s\})$, if such a point exists.
For every $i\in\{1,\ldots , k\}$ and $j\in \N$, let $L_{i,j}$ be the list of all ordered pairs $(s,q_i(s))$ with $2^{-j}\leq \|s q_i(s)\|< 2^{1-j}$,
sorted in decreasing order of the orthogonal projection of $s$ to $\overrightarrow{r}_i$
(e.g., if $\overrightarrow{r}_i$ is the $x$-axis, then $L_{i,j}$ is sorted by decreasing $x$-coordinates of the points $s$.)
For every $i\in\{1,\ldots , k\}$, let $L_i$ be the concatenation of the lists $L_{i,j}$ by decreasing $j$ (i.e., shorter edges first).
%list of all ordered pairs $(s,q_i(s))$ sorted in nonincreasing order of the length of segment $sq_i(s)$; ties are broken arbitrarily.
%(For example, if $s$ is the origin and $r_i$ is the $x$-axis, then the pairs $(s,q_i(s))$ are sorted in nondegreasing order by the $x$-coordinates of the points $q_i(s)$.)

For all $s\in S$, indices $i,j\in \{1,\ldots ,k\}$, and parameter $f\geq 0$, let $q_{i,j}(s,f)$ denote a closest point to $s$ that lies in
$C_{i,j}(s)\cap S$ and satisfies $\|s q_{i,j}(s,f)\|> f$; if such a point exists.

\smallskip\noindent\textbf{Main Phase: Computing a Spanner.}
Initialize an empty graph $G=(S,E)$ with $E:=\emptyset$.
\begin{enumerate}
\item For all $i\in \{1,\ldots , k\}$, do:
    \begin{itemize}
    \item While the list $L_i$ is nonempty, do:
        \begin{enumerate}
        \item Let $(p,q_i(p))$ be the first ordered pair in $L_i$, and put $q_i:=q_i(p)$ for short.
        \item Add (the unordered edge) $pq_i$ to $E$.
        \item For all $i'\in \{i-1,i, i+1\}$ and $j\in \{1,\ldots ,k\}$, do:\\
        Put $q_{i',j}:= q_{i',j}(p, \frac13 \|pq_i\|)$.
        If $q_{i',j}\notin B(p,q_i)$,
         then add $p q_{i',j}$ to $E$.
        \item For every $s\in \widehat{A}(p,q_i)$, including $s=p$, delete the pair $(s,q_i(s))$ from $L_i$.
        \end{enumerate}
    \end{itemize}
\item Return $G=(S,E)$.
\end{enumerate}

\begin{figure}[htbp]
\begin{center}
  \includegraphics[width=0.8\textwidth]{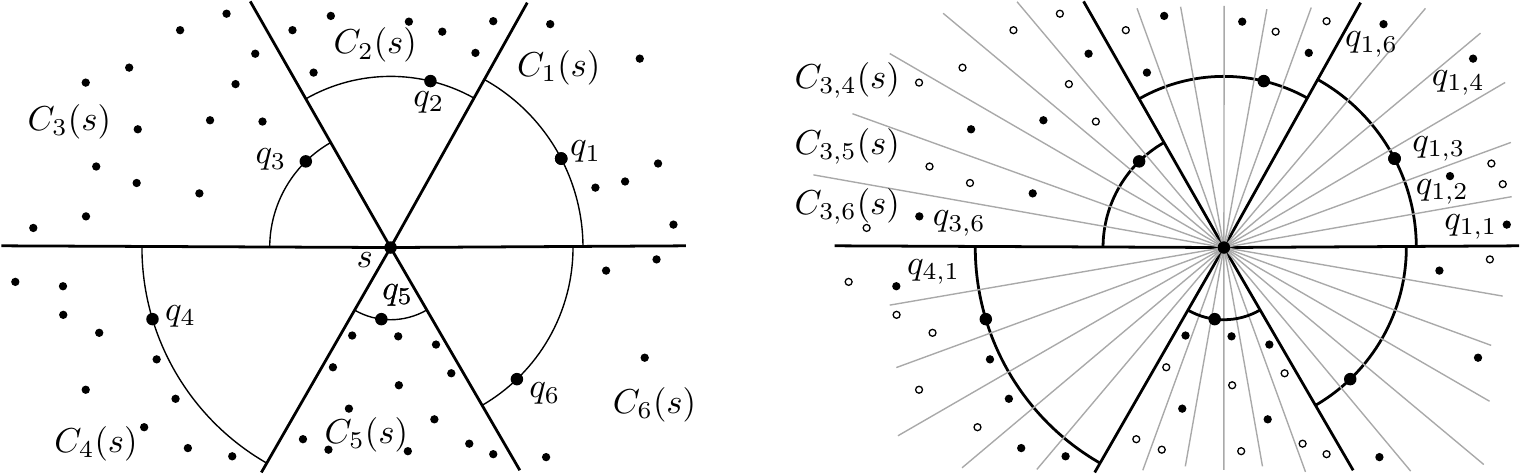}
\end{center}
\caption{Cones $C_i(s)$ and $C_{i,j}(s)$ for a point $s\in S$, with $k=6$.}
\label{fig:Yao}
\end{figure}

\begin{remark}
{\rm
Each $L_i$ is sorted by two criteria: the sublists $L_{i,j}$ are sorted by orthogonal projections of the apices $s$ to the directed lines $\overrightarrow{r}_i$,
and $L_i$ itself sorted by weight rounded down to the nearest power of two. The order by weight is used in the stretch analysis, and the order by projections in the weight analysis.

The main algorithm adds edges to the graph $G=(S,E)$ in steps~1b and~1c in each while loop.
In step~1b, it adds \emph{some} of the edges of the Yao graph $Y_k(S)$, where $k=\Theta(\eps^{-1/2})$.
It does not necessarily add all edges of $Y_k(S)$: Initially all edges of $Y_k(S)$ are in the lists $L_i$,
but some of them may be removed from $L_i$ in a step~1d, and they may not be added to $G$.

In step~1b, the algorithm adds edges from point $p$ to some points $q_{i',j}$ in the cones $C_{i',j}(p)$, where $i-1\leq i'\leq i+1$ (i.e., small cones contained in three consecutive large cones). In the large cone $C_i(p)$, we know that $q_i$ is a closest point to $p$, so for $i'=i$ the points $q_{i,j}$ are the closest points to $p$ in the small cones. However, for $i'=i\pm 1$, the points $q_{i',j}$ are not necessarily closest to $p$ in $C_{i',j}(p)$: they are the closest to $p$ among points at distance at least $\frac13 \|pq_i\|$ from $p$. Furthermore, an edge $pq_{i',j}$ is added to $G$ only if $q_{i',j}$ does not lie in the region $B(p,q_i)$ shown in Fig.~\ref{fig:wedge}. This guarantees that $\|pq_{i',j}\|$ is comparable to $\|pq_i\|$. The lower bound $\|pq_{i',j}\|\geq \frac13\|pq_i\|$ will be used in the stretch analysis, and the upper bound $\|pq_{i',j}\|\leq 2\,\|pq_i\|$ in the weight analysis.
}
\end{remark}

\begin{remark}
{\rm
It is clear that the runtime of Algorithm \textsc{SparseYao} is polynomial in $n$.
In particular, the preprocessing phase constructs the Yao-graph $Y_{k}(S)$ with $O(\eps^{-1/2}n)$ edges.
A data structure that can return $q_{i,j}(s,f)$ for a query $(i,j,s,f)$ could be based on standard range searching data structures~\cite{Aga17a}; or simply by sorting all edges $st$ by weight in each cone $C_{i,j}(s)$ in $O(n\log n)$ time. The main phase of the algorithm computes the graph $G$ in $O(\eps^{-1}n)$ time.
Optimizing the runtime, however, is beyond the scope of this paper.
}
\end{remark}

\subsection{Stretch Analysis}
\label{ssec:stretch}

In this section, we show that $G=\textsc{SparseYao}(S,\eps)$ is a $(1+\eps)$-spanner for $S$. We know that the Yao-graph with $k^2=\Theta(\eps^{-1})$ cones is a $\eps$-spanner for $S$. The following seven lemmas justify that we do not need edges in all $k^2$ cones to obtain a $(1+\eps)$-spanner.

The first three technical lemmas show that \emph{if} $G$ already contains $(1+\eps)$-paths from $a$ to $p$ and from $q$ to $b$, then we can concatenate them with the edge $pq$ to obtain a $(1+\eps)$-path from $a$ to $b$. We start with the special case that $a=p$ and $b\in B(p,q)$ in Lemma~\ref{lem:technical1}, then extend it to the case when $a\in A(p,q)$ and $b\in B(p,q)$ in Lemma~\ref{lem:technical2}. The general case, where $a\in \widehat{A}(p,q)$ and $b\in \widehat{B}(p,q)$, is tackled in Lemma~\ref{lem:technical3}. In inequalities~\eqref{eq:technical1} and~\eqref{eq:technical2} below, we use $\left(1+\frac{\eps}{2}\right)\|pq\|$ and $\left(1+\frac{\eps}{3}\right)\|pq\|$, resp., instead of $\|pq\|$ to leave room to absorb further error terms in the general case.

\begin{lemma}\label{lem:technical1}
Let $p,q\in \mathbb{R}^2$. For any two points $a=p$ and $b\in B(p,q)$, we have
\begin{equation}\label{eq:technical1}
\left(1+\frac{\eps}{2}\right)\|pq\|+(1+\eps)\|qb\|\leq (1+\eps)\|ab\|.
\end{equation}
\end{lemma}
\begin{proof}
We start with two simplifying assumptions.

\noindent (i) We may assume that $p$ is the origin, $q$ is on the positive $x$-axis, and $b$ is on or above the $x$-axis, by applying a suitable congruence if necessary.

\noindent (ii) We may assume w.l.o.g.\ that $b$ is in the boundary $\partial B(p,q)$ of $B(p,q)$, since if we rotate the segment $qb$ around $q$, then the left hand side of \eqref{eq:technical1} does not change, but the right hand side is minimized for $b\in\partial B(p,q)$.

Let us review the Taylor estimates of some of the trigonometric functions.
For the secant function, we use the upper bound
$\sec x = \frac{1}{\cos x} < 1+x^2$ for $0<x<\frac12$.
For the tangent function, we use the lower and upper bounds $x\leq \tan x<  x+\frac{x^3}{2}$  for $0<x<\frac12$.

To prove~\eqref{eq:technical1},
we distinguish between two cases based on whether $b\in \partial W_1$ or $b\in \partial W_2$.
Let $c$ be the intersection point of $\partial W_1$ and $\partial W_2$ above the $x$-axis.
Since $\angle pcq = \angle qpc = \sqrt{\eps}/4$, then $\Delta{pqc}$ is an isosceles triangle
with $\|pq\|=\|qc\|$.

\paragraph{Case~1: $b\in qc$ (Fig.~\ref{fig:cases}(left)).}
Note that $0\leq \angle qpb\leq \sqrt{\eps}/4$. Assume
$\angle qpb = t\cdot \sqrt{\eps}/4$ for some $t\in [0,1]$.
Since the interior angles of triangle $\Delta{pqb}$ add up to $\pi$,
then $\angle qbp = (2-t)\sqrt{\eps}/4$.

Let $q^\perp$ be the orthogonal projection of $q$ to $pb$. Then $\|pb\|=\|pq^\perp\|+\|q^\perp b\|$.
Since $\angle qpb \leq \angle qbp$ implies $\angle qpq^\perp  \leq \angle qb q^\perp$,
then we have $\|q^\perp b\|\leq \|p q^\perp\|$,
and so $\|q^\perp b\|\leq \frac12\, \|pb\|$.

We are now ready to prove \eqref{eq:technical1} in Case~1:
\begin{align*}
\left(1+\frac{\eps}{2}\right)\|pq\|+(1+\eps)\|qb\|
&= \left(1+\frac{\eps}{2}\right) \|pq^\perp \| \sec \angle  qpb
    +(1+\eps) \|q^\perp b\|  \sec \angle qbp \\
&= \left(1+\frac{\eps}{2}\right) \|p q^\perp\| \sec \frac{t\,\sqrt{\eps}}{4}
    +(1+\eps) \|q^\perp b\| \sec \frac{(2-t)\sqrt{\eps}}{4}\\
&\leq \left(1+\frac{\eps}{2}\right) \left(1+\frac{t^2\eps}{16}\right)\|pq^\perp\|
    +(1+\eps)\left(1+ \frac{(2-t)^2 \eps}{16}\right)\|q^\perp b\| \\
&< \left(1+\frac{(t^2+9)\eps}{16}\right)\|pq^\perp\|
    +\left(1+\frac{(t^2-4t+21)\eps}{16}\right)\|q^\perp b\| \\
&=  \left(1+\frac{(t^2+9)\eps}{16}\right) \left(\|p q^\perp\|+\|q^\perp b\|\right)
    +  \frac{(-4t+12)\eps}{16}\,  \|q^\perp b\| \\
&\leq  \left(1+\frac{(t^2+9)\eps}{16}\right) \|p b\| +  \frac{(-2t+6)\eps}{8}\cdot \frac{\|pb\|}{2}\\
&\leq  \left(1+ \eps\cdot \frac{t(t-2)+15}{16}\right) \|pb\| \\
&<  \left(1+ \eps \right) \|pb\|,
\end{align*}
as required.

\begin{figure}[htbp]
\begin{center}
  \includegraphics[width=0.9\textwidth]{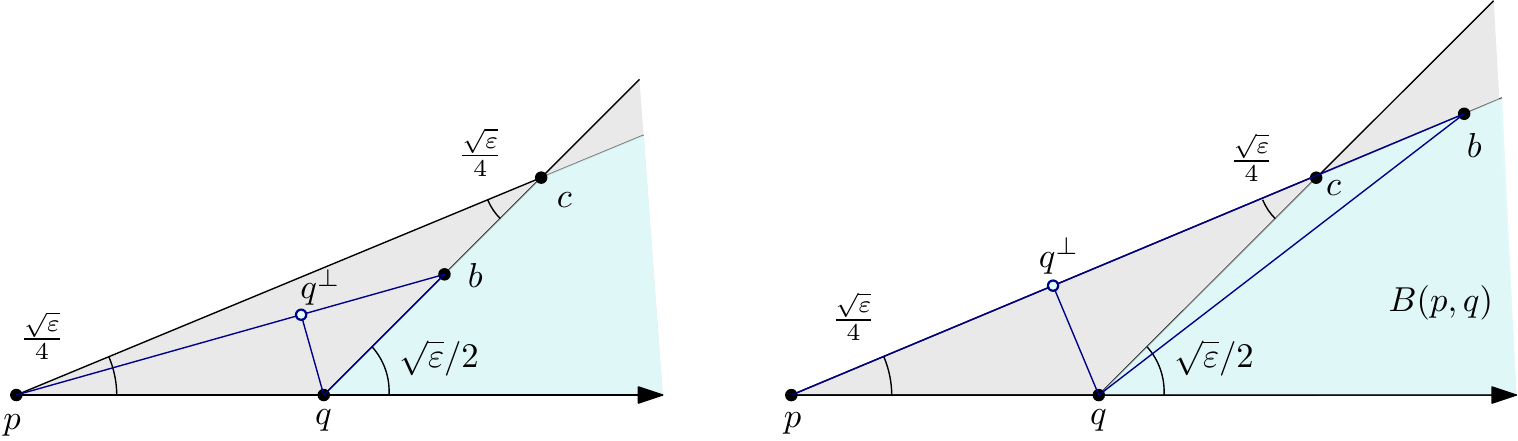}
\end{center}
\caption{Left: Case~1 where $b\in qc$.
    Right: Case~2, where $b$ lies to the right of $c$ on the boundary of $B(p,q)$.}
\label{fig:cases}
\end{figure}

\paragraph{Case~2: $b\in \partial B(p,q)\cap \partial W_1$ (Fig.~\ref{fig:cases}(right)).}
In this case, $\angle qpb=\sqrt{\eps}/4$ is fixed.
Note that $0\leq \angle pbq \leq \angle pcq\leq  \sqrt{\eps}/4$.
Assume $\angle pbq = t\cdot \sqrt{\eps}/4$ for some $t\in [0,1]$.

Since $\angle qpb \geq \angle qbp$ implies $\angle qp q^\perp \geq qb q^\perp$,
then we get $\|p q^\perp\|\leq \|q^\perp b\|$ hence $\|p q^\perp\|\leq \frac12\, \|pb\|$.
Furthermore, the right triangles $\Delta{pq q^\perp}$ and $\Delta{bq q^\perp}$ yield
 \[
 \|q q^\perp\|=\|p q^\perp\| \tan \angle qpb = \|q^\perp b\| \tan \angle qbp.
 \]
This further implies
\[
\|q^\perp b\|
=\|p q^\perp\| \frac{\tan \angle qpb}{\tan \angle qbp}
= \|p q^\perp\| \frac{\tan \left(\frac{\sqrt{\eps}}{4}\right) }{\tan\left(t\cdot \frac{\sqrt{\eps}}{4}\right)}
\leq \|pq^\perp \| \frac{\frac{\sqrt{\eps}}{4}+\frac12 \left(\frac{\sqrt{\eps}}{4}\right)^3}{t\cdot \frac{\sqrt{\eps}}{4}}
\leq \|p q^\perp \| \frac{1+\eps/32}{t}.
\]

We are now ready to prove \eqref{eq:technical1} in Case~2:
\begin{align*}
\left(1+\frac{\eps}{2}\right)\|pq\|+(1+\eps)\|qb\|
&= \left(1+\frac{\eps}{2}\right) \|pq^\perp\| \sec \angle qpb +(1+\eps) \|q^\perp b\|\sec \angle qbp \\
&=  \left(1+\frac{\eps}{2}\right) \|p q^\perp \|\sec \frac{\sqrt{\eps}}{4} +(1+\eps)\|q^\perp b\| \sec \frac{t\,\sqrt{\eps}}{4}\\
&\leq  \left(1+\frac{\eps}{2}\right)\left(1+\frac{\eps}{16}\right) \|p q^\perp \|
    +(1+\eps)\left(1+\frac{t^2\eps}{16}\right)\|q^\perp b\| \\
&<  \left(1+\frac{10\eps}{16}\right) \|p q^\perp \|
    +\left(1+\eps+\frac{t^2\eps(1+\eps)}{16}\right)\|q^\perp b\| \\
&=(1+\eps) \left(\|p q^\perp\|+\|q^\perp b\|\right) +
    \frac{\eps}{16}\left(t^2(1+\eps)\|q^\perp b\|-6 \|p q^\perp \|  \right)\\
&\leq (1+\eps) \|pb\| +  \frac{\eps}{16}\left(t(1+\eps)\left(1+\frac{\eps}{32}\right) -6 \right) \|p q^\perp \|\\
&< (1+\eps) \|pb\|,
\end{align*}
as required, since $0<t\leq 1$ and $0<\eps<1/9$.

We have confirmed \eqref{eq:technical1} in both cases. The completes the proof of Lemma~\ref{lem:technical1}.
\end{proof}

For points $p,\in \mathbb{R}^2$, let $\overline{A}(p,q)$ be the line segment of weight $\frac{\sqrt{\eps}}{2}\,\|pq\|$ on the line $pq$ with one endpoint at $p$ but interior-disjoint from the ray $\overrightarrow{pq}$.
In particular, we have $A(p,q)\subset \overline{A}(p,q)$.
%this implies that $A(p,q)$ is a line segment on the nonpositive $x$-axis; see Fig.~\ref{fig:wedge}.

\begin{lemma}\label{lem:technical2}
For all $a\in \overline{A}(p,q)$ and $b\in B(p,q)$, we have
\begin{equation}\label{eq:technical2}
(1+\eps)\|ap\|+\left(1+\frac{\eps}{3}\right)\|pq\|+(1+\eps)\|qb\|\leq (1+\eps)\|ab\|.
\end{equation}
\end{lemma}
\begin{proof}
Note that $a\in \overline{A}(p,q)$ implies $B(p,q)\subseteq B(a,q)$ since $a\in \overline{A}(p,q)$ lies on the common symmetry axis of $B(p,q)$, $W_1$, and $W_2$. Then \eqref{eq:technical1} can be written as
 \begin{align}\label{eq:reduction}
\left(1+\frac{\eps}{2}\right)\|aq\|+(1+\eps)\|qb\| &\leq (1+\eps)\|ab\| \nonumber\\
\left(1+\frac{\eps}{2}\right)\left(\|ap\| + \|pq\|\right)+(1+\eps)\|qb\|&\leq (1+\eps)\|ab\| \nonumber\\
\left(\|ap\|+\frac{\eps}{2}\, \|ap\| + \frac{\eps}{6}\, \|pq\|\right)
   + \left(1+\frac{\eps}{3}\right)\|pq\|+(1+\eps)\|qb\|&\leq (1+\eps)\|ab\|\nonumber\\
(1+\eps)\|ap\|+\left(1+\frac{\eps}{3}\right)\|pq\|+(1+\eps)\|qb\|&\leq (1+\eps)\|ab\|,
\end{align}
as $\|ap\|\leq \|\overline{A}(pq)\| = \sqrt{\eps} \|pq\| < \frac13\,\|pq\|$ for $\eps<\frac19$.
\end{proof}

In the general case, we have  $a\in \widehat{A}(p,q)$ and $b\in \widehat{B}(p,q)$.
However, for technical reasons (cf.~Lemma~\ref{lem:tilde} below), we use a larger neighborhood instead of $\widehat{A}(p,q)$.
Recall that $\widehat{A}(p,q)$ is the set of points in $\R^2$ within distance $\frac{\eps}{64}\,\|pq\|$ from $A(p,q)$.
Now let $\tilde{A}(p,q)$ be the set of points in $\R^2$ within distance at most $\frac{3\eps}{32}$ from $\overline{A}(p,q)$.

\begin{lemma}\label{lem:technical3}
For all $a\in \tilde{A}(p,q)$ and $b\in \widehat{B}(p,q)$, we have
\begin{equation}\label{eq:technical3}
(1+\eps)\|ap\|+\|pq\|+(1+\eps)\|qb\|\leq (1+\eps)\|ab\|.
\end{equation}
\end{lemma}
\begin{proof}
Since $a\in \tilde{A}(p,q)$ and $b\in \widehat{B}(p,q)$, then there exist
$a'\in \overline{A}(p,q)$ and $b'\in B(p,q)$ with $\|aa'\|\leq \frac{3\eps}{32}\,\|pq\|$ and $\|bb'\|\leq \frac{\eps}{32}\,\|pq\|$.
In particular, $\|aa'\| + \|bb'\|\leq \frac{\eps}{8}\,\|pq\|$.
%By the triangle inequality, we have $\|ap\|\leq \|aa'\|+\|a'p\|\leq \|a'p\|+\frac{\eps}{4}\,\|pq\|$ and $\|qb\|\leq \|qb'\|+\|b'b\|\leq \|qb'\|+\frac{\eps}{32}\,\|pq\|$.
Combining these inequalities with Lemma~\ref{lem:technical2} for points $a'$ and $b'$, we obtain
\begin{align*}
(1+\eps)\|ap\|+\|pq\|+(1+\eps)\|qb\|
&\leq \Big((1+\eps)\|a'p\|+\|pq\|+(1+\eps)\|qb'\|\Big) + (1+\eps)\Big(\|aa'\|+\|bb'\|\Big)\\
&\leq \left( (1+\eps)\|a'b'\| - \frac{\eps}{3}\, \|pq\|\right) + (1+\eps)\frac{\eps}{8}\, \|pq\| \\
&\leq (1+\eps)\Big( \|a'a\|+\|ab\| +\|bb'\|\Big) + \left(\frac{(1+\eps)}{8}-\frac13\right)\eps\,\|pq\| \\
&\leq (1+\eps)\|ab\| + \left((1+\eps)\frac{1}{4}-\frac13\right)\eps\,\|pq\| \\
&\leq (1+\eps)\|ab\|,
\end{align*}
for $0<\eps<1/9$, as claimed.
\end{proof}

\paragraph{Relation between $B(p,q)=W_1\cap W_2$ and $W_1\setminus W_2$.}
The following two lemmas help analyze step~1c of Algorithm \textsc{SparseYao}
that adds some of the edges $pq_{i',j}$ to the spanner.
For points $p,q\in S$, recall that $B(p,q)=W_1\cap W_2$, where $W_1$ and $W_2$ are wedges of aperture $\frac12\cdot\sqrt{\eps}$ and $\sqrt{\eps}$, resp.; see Fig.~\ref{fig:wedge}.

\begin{lemma}\label{lem:WWW}
Let $p,q\in \mathbb{R}^2$, and $B(p,q)=W_1\cap W_2$. For every $q'\in W_1\setminus W_2$, we have $\|pq'\|\leq 2\,\|pq\|$.
\end{lemma}
\begin{proof}
The line segment $pq$ decomposes $W_1\setminus W_2$ into two isosceles triangles. By the triangle inequality, the diameter of each isosceles triangle is less than $2\|pq\|$. This implies that for any point $q'\in W_1\setminus W_2$, we have $\|pq' \|< 2\, \|pq\|$.
\end{proof}

\begin{figure}[htbp]
\centering
  \includegraphics[width=0.7\textwidth]{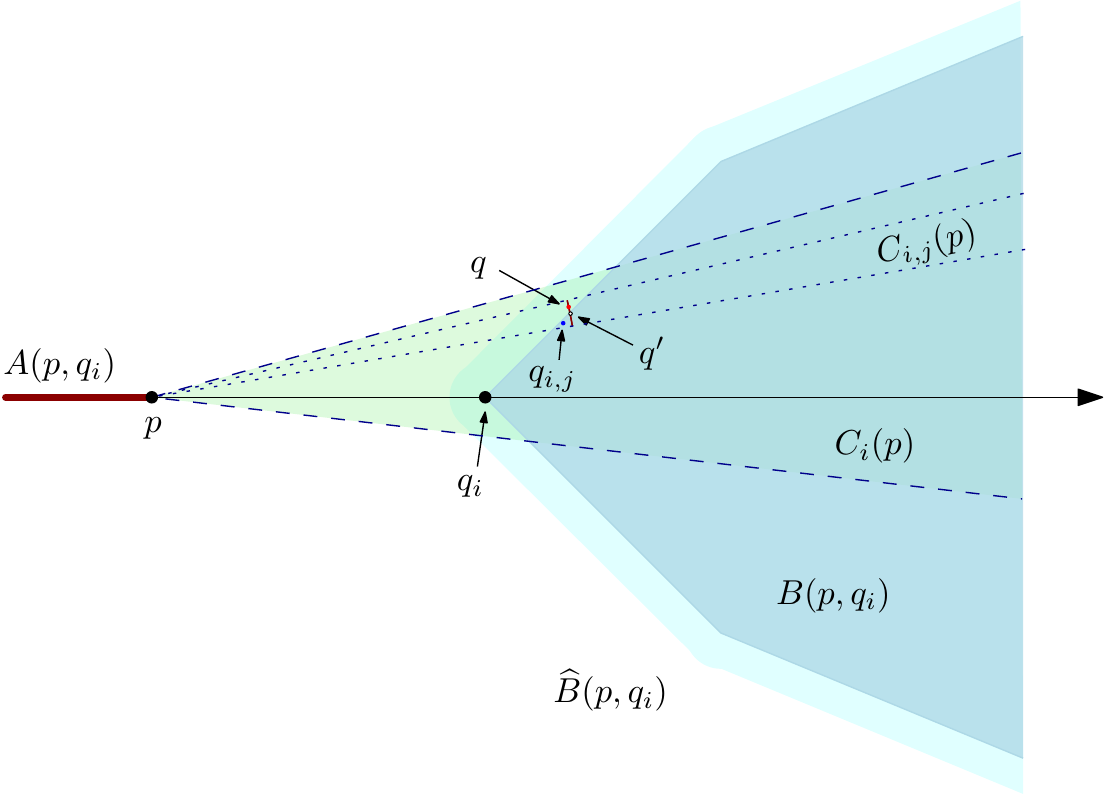}
\caption{If $q\in C_{i,j}(p)$, $q\notin B(p,q_i)$, but $q_{i,j}\in B(p,q_i)$, then $q\in \widehat{B}(p,q_i)$.}
\label{fig:wedge25}
\end{figure}

The following lemma justifies the role of the regions $\widehat{B}(p,q_i)$ in step~1d of the algorithm.

\begin{lemma}\label{lem:bbb}
Let $p,q\in S$, and assume that $q\in C_{i',j}(p)$ for some indices $i,j\in \{1,\ldots , k\}$ and $i'\in \{i-1,i ,i+1\}$, where $q_{i',j}=q_{i',j}(p,\frac13\||pq_i\|)$ is a closest point to $p$ in $C_{i',j}(p)$ at distance at least $\frac13\|pq_i\|$ from $p$.
If $q\notin {B}(p,q_i)$ but $q_{i',j}\in B(p,q_i)$, then $q\in \widehat{B}(p,q_i)$.
\end{lemma}
\begin{proof}
Since $q\in C_{i-1}(p)\cup C_i(p)\cup C_{i+1}(p)$ and the aperture of each $C_i(p)$ is
$2\pi/k\leq \frac18\cdot \sqrt{\eps}$, then $\angle qpq_i\leq \frac{1}{4}\cdot \sqrt{\eps}$,
and so $q\in W_1\setminus W_2$. Lemma~\ref{lem:WWW} yields $\|pq\|\leq 2\, \|pq_i\|$.
Both $q$ and $p_{i',j}$ are in the cone $C_{i',j}(p)$.
Consider the circle of radius $\|pq\|$ centered at $p$; see Fig.~\ref{fig:wedge25}.
Since $\|pq\|\geq \|pq_{i',j}\|$, and $q\notin {B}(p,q_i)$, $q_{i',j}\in B(p,q_i)$,
this circle intersects $\partial B(p,q_{i'})$ in the cone $C_{i',j}(p)$.
Denote by $q'$ the intersection point. Now $q,q'\in C_{i',j}(p)$ implies
that $\angle q'pq_{i',j}\leq 2\pi/k^2\leq \eps/128$.
The distance $\|qq'\|$ is bounded above by the length of the circular arc
between them:
$\|qq'\|\leq \|pq\| \angle qpq'
\leq \|pq\| \angle qpq_{i',j}
\leq \frac{\eps}{128}\,\|pq\|
\leq \frac{\eps}{128}\,2\,\|pq_i\| =\frac{\eps}{64}\|pq_i\|$.
As $q'\in B(p,q_i)$, we have $\mathrm{dist}(q,B(p,q_i))\leq \|qq'\|<\frac{\eps}{32}\|pq_i\|$,
and so $q\in \widehat{B}(p,q_{i'})$, as required.
\end{proof}

We can also clarify the relation between $\widehat{A}(p,q)$ and $\tilde{A}(p,q)$ in the setting used in the stretch analysis. Specifically, we show $\widehat{A}(p,q_i(p))\subset \tilde{A}(p,q_{i',j}(p,\frac13\|pq_i\|))$ for $i'\in \{i-1,i, i+1\}$.

\begin{lemma}\label{lem:tilde}
Let $p\in S$, and let $q_i=q_i(p)$ be a closest point to $p$ in $C_i(p)$, and $q_{i',j}=q_{i',j}(p, \frac13\|pq_i\|)$ a closest point to $p$ among all points in $C_{i',j}(p)$ at distance at least $\frac13\|pq_i\|$ from $p$, for some indices $i,j\in \{1,\ldots , k\}$ and $i'\in \{i-1,i, i+1\}$.
Then $\widehat{A}(p,q)\subset \tilde{A}(p,q_{i',j})$.
\end{lemma}
\begin{proof}
By our assumptions, we have $q_i\in C_i(p)$ and $q_{i',j}\in C_{i-1}(p)\cup C_i(p)\cup C_{i+1}(p)$.
Since the aperture of each cone $C_i(p)$ is $2\pi/k\leq \frac18\cdot \sqrt{\eps}$, then $\angle q_{i',j}pq_i\leq \frac{1}{4}\cdot \sqrt{\eps}$.
As $\frac13 \|pq_i\|\leq \|pq_{i',j}\|$, then
$$\|A(p,q_i)\| = \frac{\sqrt{\eps}}{16}\|pq_i\| \frac{3\sqrt{\eps}}{16}\|p q_{i',j} < \frac{\sqrt{\eps}}{2}\|pq_{i',j}\| = \|\overline{A}(p,q_{i',j})\|.$$
Consequently, every point in $A(p,q_i)$ is within distance at most
$$\|A(p,q_i)\| \sin \angle q_ipq_{i',j}) \leq \|A(p,q_i)\| \angle q_ipq_{i',j} \leq \frac{\sqrt{\eps}}{16}\, \|pq_i\|\cdot \frac{1}{4} \, \sqrt{\eps} = \frac{\eps}{64}\, \|pq_i\|$$
from $\overline{A}(p,q_{i',j})$. By the triangle inequality, every point in the $(\frac{\eps}{64}\, \|pq_i\|)$-neighborhood of $A(p,q_i)$ is at distance at most $(\frac{\eps}{64}+\frac{\eps}{64})\|pq_i\| = \frac{\eps}{32}\, \|pq_i\| \leq \frac{3\eps}{32}\, \|pq_{i',j}\|$
from segment $\overline{A}(p,q_{i',j})$.
\end{proof}

\begin{figure}[htbp]
\centering
  \includegraphics[width=0.98\textwidth]{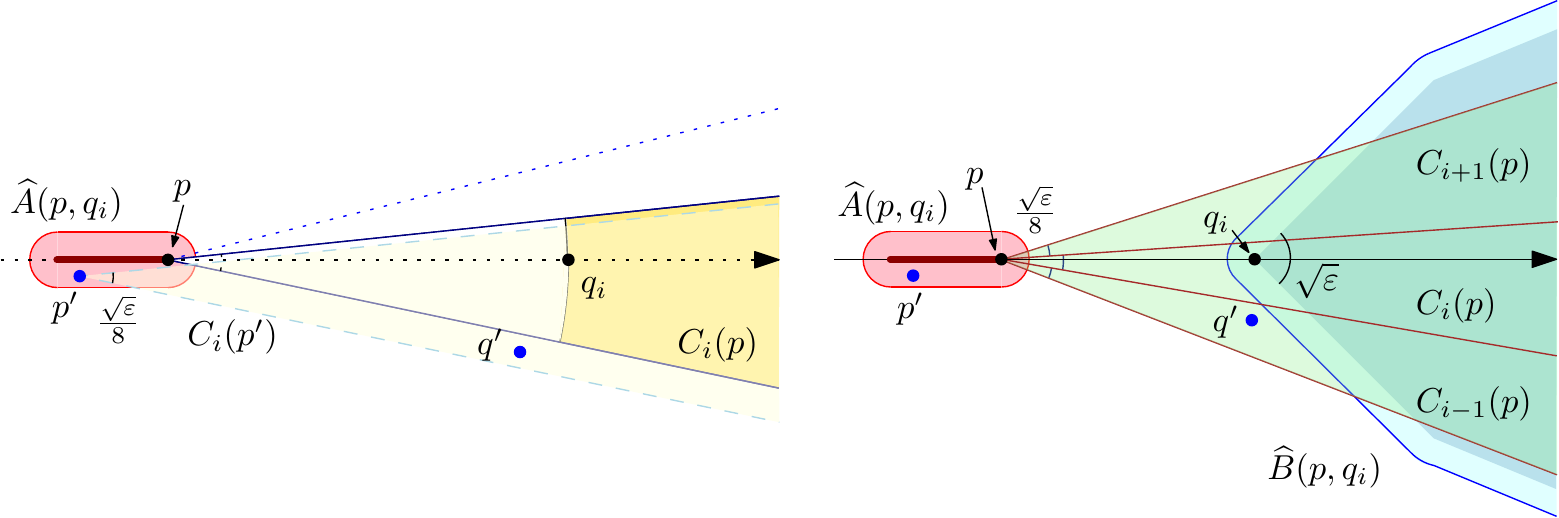}
\caption{The relative position of $pq_i$ and $p' q'$. Specifically, $p'\in \widehat{A}(p,q_i)$ and $q'\in C_i(p')$.
Left: Cones $C_i(p)$ and $C_i(p')$ with $q'\notin C_i(p)$
Right: $q'\in C_{i-1}(p)\cup C_i(p)\cup C_{i+1}(p)$ and the region $B(p'q_i')$.}
\label{fig:wedge6}
\end{figure}

Finally, we show that the cones $C_{i-1}(p)$, $C_i(p)$, and $C_{i+1}(p)$ in step~1c of Algorithm  \textsc{SparseYao}
jointly cover any point $q'\in C_i(p')$ if $p'\in  \widehat{A}(p,q_i)$ and $\|pq'\|\geq \frac12\|pq_i\|$; see Fig.~\ref{fig:wedge6}.

\begin{lemma}\label{lem:fivecones}
Let $p,q_i\in \mathbb{R}^2$ such that $q_i$ is the closest point to $p$ in the cone $C_i(p)$.
Let $p'\in \widehat{A}(p,q_i)$ and $q'\in C_i(p')$ such that $\|p'q'\|\geq \frac12 \|pq_i\|$.
Then %$q'$ lies in the union of three consecutive cones:
\begin{equation}\label{eq:fivecones}
%q'\in \bigcup_{i'=i-1}^{i+1}C_{i'}(p).
q'\in C_{i-1}(p)\cup C_{i}(p)\cup C_{i+1}(p)
\end{equation}
\end{lemma}
\begin{proof}
Recall that the aperture of cone $C_i(p)$ is $\frac{2\pi}{k}\leq \frac{1}{8}\cdot \sqrt{\eps}$.
Assume w.l.o.g.\ that $p$ is the origin, and the symmetry axis of the cone $C_i(p)$ is the is the positive $x$-axis.
Assume further that $p'=(x(p'),y(p'))$ and $q'=(x(q'),y(q'))$ in this coordinate system.
It is enough to show that the angle between segment $pq'$ and the positive $x$-axis is less than $\frac{3}{16}\cdot \sqrt{\eps}$.
In the remainder of the proof, we estimate the tangent of this angle.

Since $q'\in C_i(p')$ and $\eps\in (0,\frac19)$, the coordinates of vector $\overrightarrow{p'q'}$ are bounded by
\begin{align*}
|x(q')-x(p')| &
    \geq \|p'q'\| \cos \frac{\sqrt{\eps}}{16} \geq \|p'q'\| \cos \frac{1}{48} > \frac{63}{64}\|p'q'\|,\\
|y(q')-y(p')| &
    \leq \|p'q'\|\sin \frac{\sqrt{\eps}}{16}
    \leq \frac{\sqrt{\eps}}{16}\, \|p'q'\|.
\end{align*}
Since $p'\in \widehat{A}(p,q_i)$ and $q'\in C_i(p)$, then the coordinates of $p'$ are bounded by
\begin{align*}
|x(p')| &
    \leq \|A(pq_i)\|+\mathrm{dist}(p',A(pq_i))
    \leq \frac{\sqrt{\eps}}{16}\cdot \|pq_i\|+ \frac{\eps}{64}\|pq_i\|
   % \leq \left(\frac{1}{48}+\frac19\cdot \frac{1}{64}\right) \|pq_i\|
    \leq \frac{13}{576}\|pq_i\|
    \leq \frac{1}{32}\|p'q'\|,\\
|y(p')| &
    \leq \|A(pq_i)\|\sin \frac{\sqrt{\eps}}{16}+\mathrm{dist}(p',A(pq_i))
    \leq \frac{\sqrt{\eps}}{16}\cdot \|pq_i\|\cdot \frac{\sqrt{\eps}}{16}+ \frac{\eps}{64}\|pq_i\|
    <    \frac{\eps}{51}\|pq_i\|
    \leq \frac{\sqrt{\eps}}{17}\|p'q'\|.
\end{align*}
These bounds yield
$$ \frac{|y(q')|}{|x(q')|}
\leq \frac{|y(q')-y(p')|+|y(p')|}{|x(q')-x(p')|-|x(p')|}
< \frac{ (\frac{1}{16}+\frac{1}{17})\, \sqrt{\eps}\, \|p'q'\|}{ (\frac{63}{64}-\frac{1}{32})\, \|p'q'\|}
=\frac{132}{1037} \sqrt{\eps}< \frac{\sqrt{\eps}}{7}.
$$
Using the Taylor estimate $x\leq \tan x\leq x+\frac{x^3}{2}$ in this range, the angle between segment $pq'$ and the positive $x$-axis is less than $\frac{3}{16}\cdot \sqrt{\eps}$.
This completes the proof of~\eqref{eq:fivecones}.
\end{proof}

\paragraph{Completing the Stretch Analysis.}
We are now ready to present the stretch analysis for \textsc{SparseYao}$(S,\eps)$.

\begin{theorem}\label{thm:twostage}
For every finite point set $S\subset \R^2$ and $\eps\in (0,\frac19)$, the graph $G=\textsc{SparseYao}(S,\eps)$ is a $(1+\eps)$-spanner.
\end{theorem}
\begin{proof}
Let $S$ be a set of $n$ points in the plane.
Let $L_0$ be the list of all $\binom{n}{2}$ edges of the complete graph on $S$ in increasing order by Euclidean weight (with ties broken arbitrarily). For $\ell=1,\ldots , \binom{n}{2}$, let $e_\ell$ be the $\ell$-th edge in $L_0$, and let $E(\ell)=\{e_1,\ldots ,e_\ell\}$. We prove the following by induction on $\ell$:
    \begin{claim}\label{cl:induction}
    For every edge $ab\in E(\ell)$,  $G=(S,E)$ contains an $ab$-path of weight at most $(1+\eps)\|ab\|$.
    \end{claim}
For $\ell=1$, the claim clearly holds, as the shortest edge $pq$ is necessarily the closest points to $p$ in the cone $C_i(p)$, and so the algorithm adds $pq$ to $E$.
Assume that $1<\ell\leq \binom{n}{2}$ and Claim~\ref{cl:induction} holds for $\ell-1$.
If the algorithm added edge $e_\ell$ to $E$, then Claim~\ref{cl:induction} trivially holds for $\ell$.

Suppose that $e_\ell\notin E$. Let $e_\ell=pq$, and $q\in C_{i,j}(p)$ for some $i,j\in \{1,\ldots , k\}$.
Let $q_i=q_i(p)$ and $q_{i,j}=q_{i,j}(p,\frac13 \|pq_i\|)$ be points specified in the preprocessing phase of Algorithm \textsc{SparseYao}.
We distinguish between two cases.

\smallskip\noindent\textbf{(1) The algorithm added the edge $pq_i$ to $E$.}
Note that $\|q_i q\|\leq \|pq\|$ and $\|q_{i,j} q\| \leq \|pq\|$.
By the induction hypothesis, $G$ contains a $q_iq$-path $P_1$ of weight at most $(1+\eps)\|q_iq\|$ and a $q_{i,j}q$-path $P_2$ of weight at most $(1+\eps)\|q_{i,j} q\|$.
If $q\in \widehat{B}(p,q_i)$, then $pq_i+P_1$ is a $pq$-path of weight at most $(1+\eps)\|pq\|$ by Lemma~\ref{lem:technical3}.
Otherwise, $q\notin \widehat{B}(p,q_i)$.
In this case, $q_{i,j}\notin B(p,q_i)$ by Lemma~\ref{lem:bbb}.
This means that the algorithm added the edge $pq_{i,j}$ to $E$.
We have $q\in \widehat{B}(p,q_{i,j})$ by Lemma~\ref{lem:bbb}, and so
$pq_{i,j}+P_2$ is a $pq$-path of weight at most $(1+\eps)\|pq\|$ by Lemma~\ref{lem:technical3}.

\smallskip\noindent\textbf{(2) The algorithm did not add the edge $pq_i$ to $E$.}
Then the algorithm deleted $(p,q_i)$ from the list $L_i$ in some step~1d;
and in step~1b of the same iteration, it added another edge $p'q_i'$ to $E$,
where $q_i':=q_i'(p')$ is the closest point to $p'$ in the cone $C_i(p')$
chosen in the preprocessing phase of Algorithm \textsc{SparseYao}.
This means that $p\in \widehat{A}(p',q_i')$.
Since $L_i$ is the concatenation of the lists $L_{i,j}$ according to increasing weight,
then $\|p'q_i'\|< 2\|pq_i\|$.

By Lemma~\ref{lem:fivecones} (with the roles of $p,q_i$ and $p',q'$ interchanged), we have
$q\in C_{i-1}(p')\cup C_i(p')\cup C_{i+1}(p')$. We can now distinguish between two subcases:

\smallskip\noindent\textbf{(2a) $q\in \widehat{B}(p',q_i')$.}
By induction, $G$ contains a $pp'$-path $P_1$ of weight at most $(1+\eps)\|pp'\|$ and a $q_i'q$-path $P_2$ of weight at most $(1+\eps)\|q_i'q\|$. By Lemma~\ref{lem:technical3} (with $a=p$ and $b=q$),
the concatenation $P_1+p'q_i'+P_2$ is a $pq$-path of weight at most $(1+\eps)\|pq\|$.

\smallskip\noindent\textbf{(2b) $q\notin \widehat{B}(p',q_i')$.}
Then Lemma~\ref{lem:fivecones} implies $q\in C_{i',j'}(p')$ for some $i'\in \{i-1,i,i+1\}$ and $j'\in \{1,\ldots ,k\}$.
We claim that $\frac13\|p'q_i'\|\leq \|p'q\|$. Indeed, we have $\frac12 \|p'q_i'\|< \|pq_i\|\leq \|pq\|$ by assumption.
Since $p\in \widehat{A}(p',q_i')$ and $\eps\in (0,\frac19)$, then
$\|pp'\|
\leq \|A(p',q_i')\|+\frac{\eps}{64}\|p' q_i'\|
\leq  (\frac{\sqrt{\eps}}{16}+\frac{\eps}{64})\|p' q_i'\|
\leq \frac{13}{576}\|p' q_i'\|<\frac{1}{32}\|p'q_i'\|$.
Now the triangle inequality yields
$\|p'q\|
\geq \|qp\|-\|pp'\|
\geq (\frac12 - \frac{1}{32})\|p' q_i'\|
> \frac13 \|p'q_i'\|$.

Let $q_{i',j'}=q_{i',j'}(p',\frac13\| p' q_i(p')\|)$ be the closest point to $p'$ in the cone $C_{i',j'}(p')$ at distance at least
$\frac13\| p' q_i(p')\|$ from $p'$, as specified in the preprocessing phase of Algorithm \textsc{SparseYao}.
Since $q\notin \widehat{B}(p',q_i')$, then Lemma~\ref{lem:bbb} gives $q_{i',j'}\notin B(p',q_i')$.
Thus the algorithm added the edge $p'q_{i',j'}$ in step~1c.
Since $\|p'q_{i',j'}(p')\|\geq \frac13\|p'q_{i'}\|$, then we have $p\in \widehat{A}(p',q_i')\subset \tilde{A}(p',q_{i',j'})$ by Lemma~\ref{lem:tilde}; and $q\in \widehat{B}(p',q_{i',j'})$ by Lemma~\ref{lem:bbb}.
By induction, $G$ contains a $pp'$-path $P_1$ of weight at most $(1+\eps)\|pp'\|$, and a $q_{i',j'}q$-path $P_2$ of weight at most $(1+\eps)\|q_{i',j'}q\|$. The concatenation $P_1+p'q_{i',j'}+P_2$ is a $pq$-path of weight at most $(1+\eps)\|pq\|$ by Lemma~\ref{lem:technical3}.

This completes the proof of Claim~\ref{cl:induction}, hence Theorem~\ref{thm:twostage}.
\end{proof}

\section{Spanners in the Unit Square}
\label{sec:square}

In this section we show that, for a set $S\subset [0,1]^2$ of $n$ points and $\eps\in (0,\frac19)$, Algorithm~\textsc{SparseYao} returns a $(1+\eps)$-spanner of weight $O(\eps^{-3/2}\sqrt{n})$ (Theorem~\ref{thm:UBsqaure}).
%
%The spanner \textsc{SparseYao}$(S,\eps)$ is a subgraph of the Yao-graph $Y_{k^2}(S)$ with cones of aperture $2\pi/k^2=O(\eps)$, and so it has at most $O(\eps^{-1}n)$ edges.
Recall that for all $p\in S$ and all $i\in \{1,\ldots ,k\}$, $q_i(p)$ denotes a closest point to $p$ in the cone $C_i(p)$ of aperture $\frac18\,\sqrt{\eps}$, if it exists, $pq_i(p)$ may or may not be an edge in $G$.
Let
\[F=\big\{p q_i(p)\in E(G): p\in S, i\in \{1,\ldots , k\}\big\}.\]
We first show that the weight of the edges in $F$ approximates the weight of all other edges.

\begin{lemma}\label{lem:deltoid}
If Algorithm~\textsc{SparseYao} adds $pq_i$ and $pq_{i',j}$ to $G$ in the same iteration, then $\|pq_{i',j}\|< 2\,\|pq_i\|$.
\end{lemma}
\begin{proof}
%For short, we write $q_i=q_i(p)$ and $q_{i',j}=q_{i',j}(p)$, where $i\in \{i-1,i,i+1\}$.
Since \textsc{SparseYao} added $pq_{i',j}$ to $G$, then $q_{i',j}\notin B(p,q_i)$. Recall (cf.\ Fig.~\ref{fig:wedge}) that $B(p,q_i) = W_1\cap W_2$, where $W_1$ and $W_2$ are cones centered at $p$ and $q_i$, resp., with apertures $\frac12\, \sqrt{\eps}$ and $\sqrt{\eps}$. Since the aperture of the cone $C_{i'}(p)$ is at most $\frac18\,\sqrt{\eps}$, then $\bigcup_{i'=i-1}^{i+1}C_{i'}(p)\subset W_1$.
Consequently, $\left(\bigcup_{i'=i-1}^{i+1}C_{i'}(p)\right)\setminus B(p,q_i)\subset W_1\setminus W_2$.
Lemma~\ref{lem:WWW} gives $\|pq_{i',j}(p)\|< 2\,\|pq_i(p)\|$, as claimed.
\end{proof}

\begin{lemma}\label{lem:factor}
For $G=\textsc{SparseYao}(S,\eps)$, we have $\|G\|=O(\eps^{-1/2}) \cdot \|F\|$.
\end{lemma}
\begin{proof}
Fix $p\in S$ and $i\in \{1,\ldots , k\}$. Put $q_i=q_i(p)$, for short, and suppose that $pq_i\in E(G)$.
Consider the step of the algorithm that adds the edge $pq_i$ to $G$, together with
up to $3k=\Theta(\eps^{-1/2})$ edges of type $p q_{i',j}$, where $q_{i',j}\notin B(p,q_i)$ and $i'\in \{i-1,i ,i+1\}$.
By Lemma~\ref{lem:deltoid}, $\|pq_{i',j}\|< 2 \|pq_i\|$.
The total weight of all edges $p q_{i',j}$ added to the spanner is
\[
 \|p q_i\|+\sum_{i'=i-1}^{i+1}\sum_{j=1}^k \|pq_{i',j}\|
\leq \|pq_i\|+3k\cdot 2\, \|pq_i\|
= O(k)\, \|pq_i\|
= O\left(\eps^{-1/2}\right)\|pq_i\|.
\]
Summation over all edges in $F$ yields $\|G\|=O(\eps^{-1/2}) \cdot \|F\|$.
\end{proof}

It remains to show that $\|F\|\leq O(\eps^{-1}\sqrt{n})$. For $i=1,\ldots , k$, let
\[
F_i=\{pq_i(p)\in E(G): p\in S\},
\]
that is, the set of edges in $G$ between points $p$ and a closest point $q_i(p)$ in cone $C_i(p)$ of aperture at most $\frac18\cdot \sqrt{\eps}$. We prove that $\|F_i\|\leq O(\eps^{-1/2}\, \sqrt{n})$ in Lemma~\ref{lem:Fi} below.
Since $k=\Theta(\eps^{-1/2})$ this will immediately imply $\|F\|=\sum_{i=1}^k \|F_i\| =O(k\eps^{-1/2} \sqrt{n})=O(\eps^{-1}\sqrt{n})$.

\paragraph{Charging Scheme.}
Let $i\in \{1,\ldots , k\}$ be fixed. Assume w.l.o.g.\ that the symmetry axis of the cone $C_i$ is horizontal, and the apex is the leftmost point in $C_i$. Refer to Fig.~\ref{fig:drops}(left).
For each edge $pq_i(p)\in F_i$, let $R_i(p)$ be the intersection of cone $C_i(p)$ and the disk of radius $\|pq_i(p)\|$ centered at $p$. Note that $R_i(p)$ is a sector of the disk;
and the sectors $R_i(p)$, for all $p\in S$, are pairwise homothetic.
The sector $R_i(p)$ has three vertices: Its leftmost vertex is $p$, and the other two vertices are the endpoints of a circular arc, which have the same $x$-coordinate (by symmetry). As $q_i(p)$ is a closest point to $p$ in $C_i(p)$, then $S\cap \mathrm{int}(R_i(p))=\emptyset$.

\begin{figure}[htbp]
\begin{center}
  \includegraphics[width=0.8\textwidth]{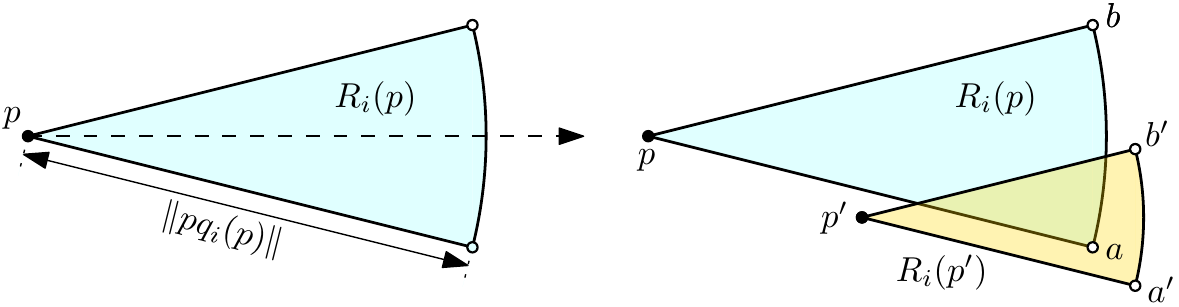}
\end{center}
\caption{Left: a sector $R_i(p)$. Right: two intersecting sectors $R_i(p)$ and $R_i(p')$.}
\label{fig:drops}
\end{figure}

Let $\mathcal{R}_i=\{R_i(p): pq_i(p)\in E(G)\}$ be the set of sectors for all edges $pq_i(p)$ in $G$.
These sectors are not necessarily disjoint; but we can still prove a lower bound on the area of their union.
We first study their intersection pattern.
\begin{lemma}\label{lem:intersect}
Assume that $R_i(p),R_i(p')\in \mathcal{R}_i$ and $R_i(p)\cap R_i(p')\neq \emptyset$.
Then $p$ and $p'$ have smaller $x$-coordinates than any other vertices of the two sectors.
\end{lemma}
\begin{proof}
Denote the vertices of $R_i(p)$ and $R_i(p')$ by $a,b,p$ and $a',b',p'$, resp.; see Fig.~\ref{fig:drops}(right).
Point $p'$ is in the exterior of $R_i(p)$, since $S\cap \mathrm{int}((R_i(p))=\emptyset$.
Now $R_i(p)\cap R_i(p')\neq \emptyset$ implies that boundaries of $R_i(p)$ and $R_i(p')$ intersect.

Suppose, to the contrary, $p'$ lies to the right of the vertical line $ab$. Since $p'$ is the leftmost point of $R_i(p')$, then all intersection points in $\partial R_i(p)\cap \partial R_i(p')$ are on the circular arc $ab$.
If $a'p'$ or $b'p'$ intersects the circular arc $ab$, then $p'\in \mathrm{int}(R_i(p))$, a contradiction.
Otherwise only the circular arcs $ab$ and $a'b'$ intersect: Then both $p$ and $p'$ lie on the orthogonal bisector
of the two intersection points; and we arrive again at a contradiction $p'\in \mathrm{int}(R_i(p))$.
\end{proof}

We partition the sectors $R_i(p)$ according to their diameters:
For all $j\in \N$, let $\mathcal{R}_{i,j}$ be the set of sectors $R_i(p)$ such that
$2^{-j}\leq \|p q_i(p)\| < 2^{1-j}$.
We show that the sectors $\mathcal{R}_{i,j}$ do not overlap too heavily.

\begin{lemma}\label{lem:density}
For all $j\in \N$, any point $g\in [0,1]^2$ is contained in $O(\eps^{-1/2})$ sectors in $\mathcal{R}_{i,j}$.
\end{lemma}
\begin{proof}
Let $\mathcal{R}_{i,j}(g)=\{R\in \mathcal{R}_{i,j}: g\in R\}$ be the set of sectors in $\mathcal{R}_{i,j}$ that contain $g$; these sectors pairwise intersect. By Lemma~\ref{lem:intersect}, the leftmost vertices of the sectors have smaller $x$-coordinates than any other vertices (i.e., endpoints of circular arcs). Let $\ell$ be a vertical line that separates the leftmost vertices of these sectors from all other vertices; and let $\ell^-$ be the left halfplane bounded by $\ell$; see Fig.~\ref{fig:terrain}(left).

\begin{figure}[htbp]
\begin{center}
  \includegraphics[width=0.98\textwidth]{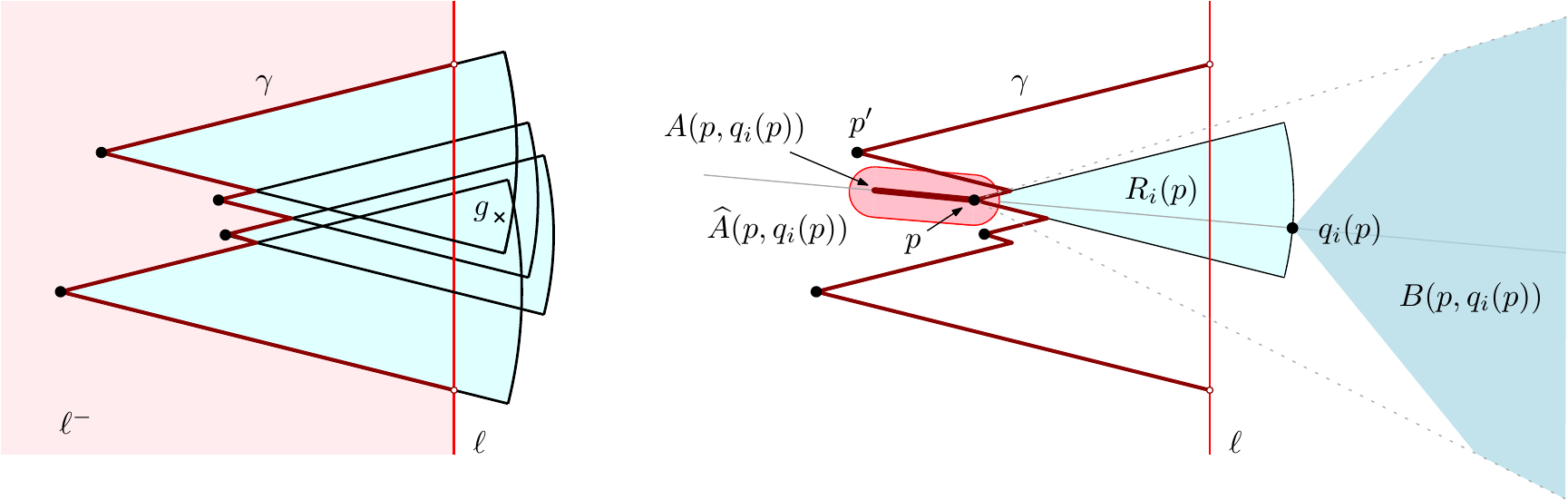}
\end{center}
\caption{Left: the set of sectors that contains point $g$; and the vertical line $\ell$.
Right: regions $A(p,q_i(p))$ and $\widehat{A}(p,q_i(p))$ for a vertex $p\in \gamma$.}
\label{fig:terrain}
\end{figure}

Recall that for every sector $R_i(p)\in \mathcal{R}_{i,j}(g)$, we have $\|pg\|\leq \|pq_i(p)\|<2^{1-j}$. As the aperture of $C_i$ is at most $\frac18\, \sqrt{\eps}$, with $\eps\in (0,\frac19)$, then the weight of the vertical segment $\ell\cap R_i(p)$ is at most
$\|\ell\cap R_i(p)\| \leq 2^{1-j}\ \cdot 2\sin (\frac{1}{16}\, \sqrt{\eps})\leq O(2^{-j} \sqrt{\eps})$.

For every sector $R_i(p)\in \mathcal{R}_{i,j}(g)$, the region $R_i(p)\cap \ell^-$ is an isosceles triangle with two legs of slopes $\pm \tan (\frac{1}{16}\sqrt{\eps})$. Consider the union of these isosceles triangles, $(\bigcup\{ R: R\in \mathcal{R}_{i,j}(g)\})\cap \ell^-$. Its right boundary is of a vertical segment of weight $O(2^{-j}\sqrt{\eps})$ along $\ell$; and its left boundary is a $y$-monotone curve, that we denote by $\gamma$. All local $x$-minima of $\gamma$ are leftmost vertices of the sectors $R_i(p)\in \mathcal{R}(g)$.

We claim that if $p,p'\in S$ are two consecutive local $x$-minima along $\gamma$, then
\begin{equation}\label{eq:height}
    |y(p)-y(p')|\geq \frac{\eps}{256}\cdot 2^{-j},
\end{equation}
that is the $y$-coordinates of $p$ and $p'$ differ by at least $\eps \cdot 2^{-(j+9)}$; refer to Fig.~\ref{fig:terrain}(right). Suppose, to the contrary, that $|y(p)-y(p')|< \eps\cdot 2^{-(j+9)}$.
Due to the slopes of $\gamma$, this implies
\[
|x(p)-x(p')|
=\frac{|y(p)-y(p')|}{|\mathrm{slope}(pp')|}
<\eps \cdot 2^{-(j+9)}{\tan \left(\frac{1}{16}\sqrt{\eps}\right)}
\leq \frac{\eps\cdot 2^{-(j+9)}}{\frac{1}{16}\cdot\sqrt{\eps}}
= \sqrt{\eps}\,2^{-(j+5)}.
\]

Assume w.l.o.g.\ that Algorithm~\textsc{SparseYao} added edge $pq_i(p)$ to $E(G)$ before $p' q_i(p')$.
As the list $L_i$ is sorted by decreasing $x$-coordinates, then $x(p')\leq x(p)$.
When the algorithm added edge $p q_i(p)$ to $E(G)$, it deleted all pairs $(s, q_i(s))$ from $L_i$ such that $s\in \widehat{A}(p,q_i(p))$. Given that $\|p q_i(p)\|\geq 2^{-j}$ and $q_i(p)\in C_i(p)$,
then $A(p,q_i(p))$ is a line segment of weight
$\|A(p,q_i(p))\| = \frac{\sqrt{\eps}}{16}\, \|p q_i(p)\| \geq 2^{-(j+4)}\, \sqrt{\eps}$
and slope of absolute value at most $\tan(\frac{1}{16}\,\sqrt{\eps})$.
This implies that $x(p')$ lies in the $x$-projection of
$A(p,q_i(p))$. Furthermore, the distance between $p'$ and the point with the same $x$-coordinate
in $A(p,q_i(p))$ is at most $2\, |y(p)-y(p')|< \eps\cdot 2^{-(j+8)}$.
Recall that $\widehat{A}(p,q_i(p))$ contains every point in an $(\frac{\eps}{64}\cdot 2^{-j})$-neighborhood of 
$A(p,q_i(p))$.  Consequently, $p'\in  \widehat{A}(p,q_i(p))$,
and the algorithm deleted $(p',q_i(p)')$ from $L_i$ when it added $p q_i(p)$ to $E(G)$.
This contradicts the assumption $p'q_i(p')\in E(G)$, and proves the claim~\eqref{eq:height}.

As the height of $\gamma$ is $O(2^{-j} \sqrt{\eps})$. Combined with \eqref{eq:height}, this implies that $\gamma$ has at most $O(2^{-j}\sqrt{\eps})/ (\frac{\eps}{64}\cdot 2^{-j}) =O(\eps^{-1/2})$ local $x$-minima. Thus
$|\mathcal{R}(g)|\leq O(\eps^{-1/2})$, as claimed.
\end{proof}

The sectors in $\mathcal{R}_{i,j}$ are not necessarily disjoint. In order to obtain disjoint regions, we define the \emph{core} of a sector $R_i(p)$, denoted $\widehat{R}_i(p)$; see Fig.~\ref{fig:core}(left). Label the vertices of $R_i(p)$ by $a$, $b$, and $p$, where $\|pa\|=\|pb\|=\|pq_i(p)\|$. Let $\mathbf{v}_i(p)$ be the vector along the angle bisector of $\angle apb$ of weight $\|\mathbf{v}_i(p)\|=\frac23 \|p q_i(p)\|$. Now $\frac13 R_i(p)$ denote the region obtained by scaling $R_i(p)$ from center $p$ with ratio $\frac13$; and let $\widehat{R}_i(p)=\frac13 R_i(p)+\mathbf{v}_i(p)$. By construction, we have $\widehat{R}_i(p)\subset R_i(p)$ and $\area(\widehat{R}_i(p))=\frac19\,\area(R_i(p))$.

\begin{figure}[htbp]
\begin{center}
  \includegraphics[width=0.8\textwidth]{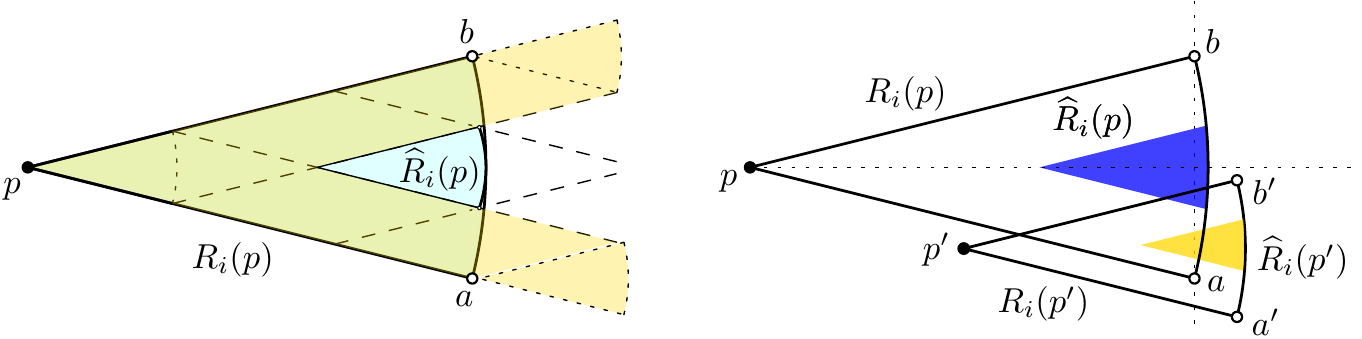}
\end{center}
\caption{Left: the core $\widehat{R}_i(p)$ of a sector $R_i(p)$.
Right: The cores $\widehat{R}_i(p)$ and $\widehat{R}_i(p')$ are disjoint.}
\label{fig:core}
\end{figure}

\begin{lemma}\label{lem:disjoint}
If $j+2\leq j'$, then any two sectors in $\mathcal{R}_{i,j}$ and $\mathcal{R}_{i,j'}$ have disjoint cores.
\end{lemma}
\begin{proof}
Let $R_i(p)\in \mathcal{R}_{i,j}$ and $R_i(p')\in \mathcal{R}_{i,j'}$ with $j+2\leq j'$. Label their vertices
by $a$, $b$, $p$ and $a'$, $b'$, $p'$, resp., in counterclockwise order; see Fig.~\ref{fig:core}(right). Note that
\[
\|a'p'\| \leq 2^{-j'} \leq 2^{-(j+2)} \leq \frac12\,\|ap\|.
\]
Suppose, for the sake of contradiction, that $\widehat{R}_i(p)\cap \widehat{R}_i(p')\neq \emptyset$.
Since $\widehat{R}_i(p)\subset R_i(p)$ and $\widehat{R}_i(p')\subset R_i(p')$, then $R_i(p)\cap R_i(p')\neq \emptyset$. By Lemma~\ref{lem:intersect}, $p'$ lies to the left of the vertical line $ab$.
Given that $p'\in S$, the point $p'$ cannot be in the interior of $R_i(p)$, that is, $p'\notin \mathrm{int}(R_i(p))$. The combination of these two observation yields $p'\notin \mathrm{int}(C_i(p))$.

We may assume that $p'$ lies below the line $pa$ by a reflection in the $x$-axis, if necessary.
Since the sector $R_i(p')$ has a horizontal symmetry axis and
$\|\mathbf{v}_i(p')\|=\frac23 \|p'a'\|= \frac23 \|p'b'\|$, then
$R_i(p')$ lies below the line $p'a'+3\mathbf{v}_i(p')$;
and the core $\widehat{R}_i(p')$ lies below the line $p'a'+2\mathbf{v}_i(p')$.
Recall that the vectors $\mathbf{v}_i(p)$ and $\mathbf{v}_i(p')$ are parallel,
and $\|a'p'\| \leq \frac12\,\|ap\|$ implies $\|\mathbf{v}_i(p')\|\leq \frac12 \|\mathbf{v}_i(p)\|$.
Consequently, $\widehat{R}_i(p')$ lies below the line $pa+\mathbf{v}_i(p)$.
However, by construction, $\widehat{R}_i(p)$ lies above the line $pa+\mathbf{v}_i(p)$.
This implies that the cores $\widehat{R}_i(p)$ and $\widehat{R}_i(p')$ are disjoint.
\end{proof}

The combination of Lemmas~\ref{lem:density} and~\ref{lem:disjoint} gives the following bound on the total area of all sectors of a given direction.

\begin{lemma}\label{lem:volume}
For every $i$, we have $\sum_{R\in \mathcal{R}_i}\area(R)=\sum_{j\in \N} \sum_{R\in \mathcal{R}_{i,j}} \area(R)\leq O(\eps^{-1/2})$.
\end{lemma}
\begin{proof}
For every sector $R\in \mathcal{R}_i(p)$, we have $\area(R)=\Theta(\area(\widehat{R}))$.
Define the function $f_{i,j}:[0,1]^2\rightarrow \N$ such that for all $g\in [0,1]^2$,
$f(g)$ is the number of cores $\widehat{R}$, $R\in \mathcal{R}_{i,j}$, that contain $g$
Then $\sum_{R\in \mathcal{R}_{i,j}} \area(\widehat{R}) =\int_{[0,1]^2} f_{i,j}(g)$.
By Lemma~\ref{lem:density}, we have $f(g)\leq O(\eps^{-1/2})$ for all $g\in [0,1]$.
By Lemma~\ref{lem:disjoint}, $\sum_{j\in\N} f_{i,3j+\ell}(g) \leq O(\eps^{-1/2})$
for all $\ell\in \{0,1,2\}$ and $g\in [0,1]$. Consequently,
\begin{align*}
\sum_{R\in \mathcal{R}_i}\area(R)
&\leq O\left(\sum_{R\in \mathcal{R}_i}\area(\widehat{R}) \right)
= O\left(\sum_{\ell=0}^2 \sum_{j\in \N} \sum_{R\in \mathcal{R}_{i,3j+\ell}} \area(\widehat{R}) \right)\\
&= O\left(\sum_{\ell=0}^2 \sum_{j\in \N} \int_{[0,1]^2} f_{i,3j+\ell}(g) \right)
= O\left(\sum_{\ell=0}^2 \int_{[0,1]^2}  \sum_{j\in \N} f_{i,3j+\ell}(g) \right)\\
&\leq O\left(\sum_{\ell=0}^2 \int_{[0,1]^2} \eps^{-1/2} \right)
= O\left(\sum_{\ell=0}^2 \eps^{-1/2}\area([0,1]^2)\right) =O(\eps^{-1/2}),
\end{align*}
as claimed.
\end{proof}

\begin{lemma}\label{lem:Fi}
For every $i\in \{1,\ldots, k\}$, we have $\|F_i\|\leq O(\eps^{-1/2}\sqrt{n})$.
\end{lemma}
\begin{proof}
For every sector $R_i(p)\in \bigcup_{j\in \N}\mathcal{R}_{i,j}$, we have
\[
\area(\widehat{R}_i(p))
\geq \Omega(\area(R_i(p)))
\geq \Omega(\|p q_i(p)\|^2 \sqrt{\eps}).
\]
In particular, summation over all edges $e\in F_i$ and Jensen's inequality gives
\[
\sum_{R\in \mathcal{R}_i} \area(R)
\geq \Omega\left(\sqrt{\eps}\, \sum_{e\in F_i} \|e\|^2\right)
\geq \Omega\left(\sqrt{\eps}\cdot |F_i|\left(\frac{1}{|F_i|}\sum_{e\in F_i} \|e\|\right)^2\right)
\geq \Omega\left(\sqrt{\eps}\cdot |F_i| \cdot w_i^2\right),
\]
where $w_i=\frac{1}{|F_i|}\sum_{e\in F_i} \|e\| = \|F_i\|/|F_i|$ is the average weight of an edge in $F_i$.
Combined with Lemma~\ref{lem:volume}, we obtain
\begin{align}
\sqrt{\eps}\cdot |F_i|\cdot w_i^2 &\leq O(\eps^{-1/2})\label{eq:key}\\
w_i^2 &\leq O\left(\frac{1}{\eps\, |F_i|}\right)\nonumber \\
w_i &\leq O\left(\frac{1}{\sqrt{\eps\, |F_i|}}\right). \nonumber
\end{align}
\later{
\begin{equation}\label{eq:key}
\sqrt{\eps}\cdot |F_i|\cdot w_i^2 \leq O(\eps^{-1/2})
\hspace{3mm}\Rightarrow\hspace{3mm}
w_i^2 \leq O\left(\frac{1}{\eps\, |F_i|}\right)
\hspace{3mm}\Rightarrow\hspace{3mm}
w_i \leq O\left(\frac{1}{\sqrt{\eps\, |F_i|}}\right).
\end{equation}
}
Finally, $\|F_i\|\leq |F_i|\cdot w_i
\leq O(|F_i| /\sqrt{\eps\, |F_i|})
= O(\eps^{-1/2}\sqrt{|F_i|})
\leq  O(\eps^{-1/2}\sqrt{n})$, as required.
\end{proof}

\begin{theorem}\label{thm:UBsqaure}
For every set of $n$ points in $[0,1]^2$ and every $\eps\in (0,\frac19)$, Algorithm \textsc{SparseYao} returns a Euclidean $(1+\eps)$-spanner of weight $O(\eps^{-3/2}\,\sqrt{n})$.
\end{theorem}
\begin{proof}
Let $G=\textsc{SparseYao}(S,\eps)$, and define $F\subset E(G)$ and $F_1,\ldots, F_k$ as above.
By Lemma~\ref{lem:Fi}, $\|F\|=\sum_{i=1}^k \|F_i\|= O(k\,\eps^{-1/2}\,\sqrt{n}) =O(\eps^{-1}\,\sqrt{n})$.
By Lemma~\ref{lem:factor}, $\|G\|\leq O(\eps^{-1/2})\cdot (\|F\|+\sqrt{2})\leq O(\eps^{-3/2}\,\sqrt{n})$,
as claimed.
\end{proof}

\section{Generalization to Higher Dimensions}
\label{sec:d-space}

\paragraph{Upper Bound.}
For every constant $d\geq 2$, Algorithm \textsc{SparseYao} and its analysis generalize to $\R^d$.
We sketch the necessary adjustments for a point set $S\subset [0,1]^d$.
Recall that for $d=2$, we partitioned the plane into $k=\Theta(\eps^{-1/2})$ cones $C_1,\ldots , C_k$, of aperture $\frac18\,\sqrt{\eps}$.
In $d$-dimensions, we can \emph{cover} $\R^d$ with $k=\Theta(\eps^{(1-d)/2})$ cones of aperture $\frac18\,\sqrt{\eps}$ such that every point in $\R^d$ is covered by at most $O(d)$ cones. With these cones, Algorithm \textsc{SparseYao} and its stretch analysis goes through almost verbatim. We point out a few dimension-dependent details: In step~1c of Algorithm~\textsc{SparseYao}, instead of three consecutive cones  $\bigcup_{i'=i-1}^{i+1}C_{i'}(p')$, we need $O(2^d)$ cones that cover a $\frac{\sqrt{\eps}}{4}$-neighborhood of $C_i(p)$, in spherical distance with respect to $\mathbb{S}^{d-1}$. For the stretch analysis, the key technical lemmas directly generalize to $d$-space: In Lemmas~\ref{lem:technical1}---\ref{lem:tilde}, the points $a$, $b$, $p$, and $q$ are coplanar in $\R^d$; and in Lemma~\ref{lem:technical3} and~\ref{lem:tilde}, the regions $\widehat{A}(p,q)$ and $\widehat{B}(p,q)$ are small neighborhoods of $A(p,q)$ and $B(p,q)$, resp., independent of dimension.

For the weight analysis of the spanner $G=\textsc{SparseYao}(S,\eps)$  also generalizes.
Standard volume argument shows that every cone of aperture $\frac18\,\sqrt{\eps}$ is covered by $O(\eps^{(1-d)/2})$ cones of aperture $\Theta(\eps^{-1})$. Thus the direct generalization of Lemma~\ref{lem:factor} yields $\|G\|=O_d(\eps^{(1-d)/2})\cdot \|F\|$, where $F$ is partitioned into $k=\Theta_d(\eps^{(1-d)/2})$ subsets $F=\bigcup_{i=1}^k F_i$.

In the generalization of Lemma~\ref{lem:density}, every generic point $g\in [0,1]^d$ is contained in $O(\eps^{(1-d)/2})$ regions $R_i(p)$: In the proof, however, the $y$-monotone curve $\gamma$ is replaced by a $(d-1)$-dimensional surface/terrain.
In the weight analysis for $\|F_i\|$, we need to charge the weight of each edge $e\in F_i$ to an empty sector $\widehat{R}_i$, which is the intersection of a cone $C_i$ of aperture $\frac18\, \sqrt{\eps}$ and a ball of radius $\|e\|$. The volume of such a region is $\Theta_d(\|e\|^d\cdot \eps^{(d-1)/2})$. The \emph{core} of the sectors can be defined analogously, and Lemmas~\ref{lem:disjoint}--\ref{lem:volume} extend to $d$-space. Finally, in the proof of Lemma~\ref{lem:Fi}, Jensen's inequality is used for the function $x\rightarrow x^d$. In particular, for the average weight of an edge in $F_i$,  $w_i=\|F_i\|/ |F_i|$, inequality~\eqref{eq:key} becomes
\begin{align}
\eps^{(d-1)/2}\cdot |F_i|\cdot w_i^d &\leq O_d(\eps^{(1-d)/2})\label{eq:key+d}\\
w_i^d &\leq O_d\left(\frac{1}{\eps^{d-1} |F_i|}\right)\nonumber \\
w_i &\leq O_d\left(\frac{1}{\eps^{1-1/d}\, |F_i|^{1/d}}\right), \nonumber
\end{align}
and $\|F_i\| =|F_i|\cdot w_i \leq O_d( \eps^{1/d-1}\,|F_i|^{1-1/d}) \leq  O_d( \eps^{1/d-1} n^{1-1/d})$. Overall,
\[
\|G\|
\leq O\left(\eps^{(1-d)/2}\right)\cdot \sum_i \|F_i\|
\leq O_d\left(\eps^{(1-d)/2}\cdot \eps^{(1-d)/2}\cdot \eps^{1/d-1} n^{1-1/d}\right)
\leq O_d\left(\eps^{-d+1/d} n^{1-1/d}\right).
\]

\paragraph{Lower Bound.}
The empty ellipse condition and the lower bound construction readily generalize to every dimension $d\geq 2$. Let $S_0$ be a set of $2m$ points, where $m=\lfloor (\eps/d)^{1-d}\rfloor$, with $m$ points arranged in a grid on two opposite faces of a unit cube $[0,1]^d$.
By the empty ellipsoid condition, every $(1+\eps)$-spanner for $S_0$ contains a complete bipartite graph $K_{m,m}$, of weight $\Omega_d(\eps^{2(1-d)})$.
if we arrange $\Theta_d(\eps^{d-1}n)$ translated copies of $S_0$ in a $(\eps^{d-1}n)^{1/d}\times \ldots \times (\eps^{d-1}n)^{1/d}$ grid, we obtain a set $S$ of $\Theta(n)$ points, and a lower bound of $\Omega_d(\eps^{1-d} n)$. Scaling by a factor of $(\eps^{d-1}n)^{1/d}$ yields a set of $\Theta(n)$ points in $[0,1]^d$ for which any $(1+\eps)$-spanner has weight $\Omega_d(\eps^{-d+1/d}\, n^{1-1/d})$.

\section{Spanners for the Integer Lattice}
\label{sec:grid}

We briefly review known results from analytic number theory in Section~\ref{ssec:Farey}, and derive upper bounds on the minimum weight of a $(1+\eps)$-spanner for the $n\times n$ grid: First we analyze the weight of  Yao-graphs (Section~\ref{ssec:gridUB}), and then refine the analysis for Sparse Yao-graphs in (Section~\ref{ssec:next}).

\subsection{Preliminaries: Farey Sequences}
\label{ssec:Farey}

Two points in the integer lattice $p,q\in \Z^2$ are \emph{visible} if the line segment $pq$ does not pass through any lattice point. An integer point $(i,j)\in \Z^2$ is visible from the origin $(0,0)$ if $i$ and $j$ are relatively prime, that is, $\mathrm{gcd}(i,j)=1$. The \emph{slope} of a segment between $(0,0)$ and $(i,j)$ is $j/i$. For every $n\in \N$, the \emph{Farey set of order $n$},
    \[ F_n=\left\{\frac{a}{b} : 0\leq a\leq b\leq n\right\},  \]
is the set of slopes of the lines spanned by the origin and lattice points $(b,a)\in [0,n]^2$ with $a\leq b$.
The \emph{Farey sequence} is the sequence of elements in $F_n$ in increasing order.
Note that $F_n\subset [0,1]$. Farey sets and sequences have fascinating properties,
and the distribution of $F_n$, as $n\rightarrow \infty$ is not fully understood.
It is known that
\[
|F_n|=1+\sum_{1\leq i\leq n} \varphi(i) = \frac{3n^2}{\pi^2}+O(n \log n),
\]
where $\varphi(i)$  is Euler's totient function (i.e., $\varphi(i)$ is the number of nonnegative integers $j$ with $1\leq j\leq i$ and $\mathrm{gcd}(i,j)=1$). Furthermore, if $\frac{p_1}{q_1}$ and $\frac{p_2}{q_2}$ are consecutive terms of the Farey sequence in reduced form (i.e.,  $\mathrm{gcd}(p_1,q_1)=1$ and $\mathrm{gcd}(p_2,q_2)=1$), then $|p_1q_2-p_2q_1|=1$ \cite{HW79}. The Farey sequence is uniformly distributed on $[0,1]$ in the sense that for any fixed subinterval $[\alpha,\beta]\subset [0,1]$, the asymptotic frequency of the Farey set in $[\alpha,\beta]$ is known converge as $n$ tends to infinity~\cite{Dress99}:
    \[ \frac{|F_n\cap [\alpha,\beta]|}{|F_n|} = \beta-\alpha +o(1). \]
The error term is bounded by $O(n^{-1+\varepsilon})$, for some $\eps>0$. However, determining the rate of convergence is known to be equivalent to the Riemann hypothesis~\cite{Franel24,Landau24}; see also~\cite{Ledoan18} and references therein.

The key result we use is a bound on the average distance to a Farey set $F_n$.
For every $x\in [0,1]$, let
    \[\rho_n(x)= \min_{\frac{p}{q}\in F_n} \left| \frac{p}{q} - x\right| \]
denote the distance between $x$ and the Farey set $F_n$.
Kargaev and Zhigljavsky~\cite{KZ96} proved that
    \begin{equation}\label{eq:KZ}
    \int_0^1 \rho_n(x) \, dx  = \frac{3}{\pi^2} \, \frac{\ln n}{n^2} + O\left(\frac{1}{n^2}\right), \hspace{5mm}\textrm{as} \hspace{5mm} n\rightarrow \infty.
    \end{equation}

\subsection{The Weight of Yao-Graphs for the Grid}
\label{ssec:gridUB}

\begin{lemma}\label{lem:num}
For a positive integer $k\in \N$, consider the subdivision of the unit interval $[0,1]$ into $k$ subintervals,
$[0,1]=\bigcup_{i=1}^k [\frac{i-1}{k},\frac{i}{k}]$. For every $i=1,\ldots , k$, let $q_i$ be the
smallest positive integer such that $\frac{i-1}{k}\leq \frac{p_i}{q_i}\leq \frac{i}{k}$ for some integer $p_i$.
Then
\begin{enumerate}
\item[{\rm (i)}] $\sum_{i=1}^k q_i =O(k^{3/2}\log^{1/2} k)$, and
\item[{\rm (ii)}] $\sum_{i=1}^k q_i^3 =O(k^3\log k)$.
\end{enumerate}
\end{lemma}
\begin{proof}
Let $[\alpha,\beta]\subset [0,1]$ be an interval of length $\beta-\alpha=\frac{1}{k}$.
If $[\alpha,\beta]$ contains a point $\frac{p_i}{q_i}\in F_n$, then $q_i\leq n$.
Otherwise, $[\alpha,\beta]$ is disjoint from the Farey set $F_n$, and
then $\rho_n(x)\geq \min\{|x-\alpha|,|x-\beta|\}$ for all $x\in [\alpha,\beta]$.
In the latter case,
\begin{equation}\label{eq:blip}
    \int_{[\alpha,\beta]} \rho_n(x) \, dx
  \geq \int_{[\alpha,\beta]}\min\{|x-\alpha|,|x-\beta|\} \, dx
  = \frac{1}{4\cdot |\beta-\alpha|^2}
  = \frac{1}{4k^2}.
\end{equation}

For every positive integer $n$, let $a(n)$ be the number of intervals in $\{[\frac{i-1}{k},\frac{i}{k}]: i=1,\ldots , k\}$ that are disjoint from $F_n$. Note that $a(n)=0$ for $n\geq k$, since all $k$ (closed) intervals contain a rational of the form $\frac{p}{k}$. In particular, we have $q_i\leq k$ for all $i=1,\ldots , k$.

The combination of \eqref{eq:KZ} and \eqref{eq:blip} yields $a(n)\leq O(k^2 \log n / n^2)$.
If we set $n=\sqrt{ck\log k}$ for a sufficiently large constant $c\geq 2$, then
    \[ a \left(\sqrt{ck\log k}\right)
    \leq O \left( \frac{k^2 \log(\sqrt{ck \log k})}{ck\log k} \right)
    = O \left( \frac{k}{2} \cdot  \frac{\log(ck\log k)}{c\log k} \right)
    = O \left( \frac{k}{2} \cdot  \frac{\log(c) + 2\log(k)}{c\log k} \right)
    \leq \frac{k}{2}.
    \]
That is, at most $\frac{k}{2}$ of $k$ intervals are disjoint from $F_n$. For the remaining at least $\frac{k}{2}$ intervals, we have $q_i\leq n\leq O(\sqrt{k\log k})$, and the sum of these $q_i$ terms is $O(k^{3/2}\log^{1/2}k)$. It remains to bound the sum of $q_i$ terms in the intervals disjoint form $F_n$. We use a standard diadic partition. Let $A=\sqrt{ck\log k}$. Then
\begin{align*}
\sum_{i=1}^k q_i
  &=O(k^{3/2}\log^{1/2} k)+\sum_{j=\lfloor \log A\rfloor}^{\lfloor \log k\rfloor} \big(a(2^j) - a(2^{j-1})\big)2^j \\
  &=O(k^{3/2}\log^{1/2} k)+\sum_{j=\lfloor \log A\rfloor}^{\lfloor \log k\rfloor} a(2^j) \big(2^j-2^{j-1}\big)\\
  &\leq O(k^{3/2}\log^{1/2} k)+\sum_{j=\lfloor \log A\rfloor}^{\lfloor \log k\rfloor} O\left(\frac{k^2 j}{2^{2j}}\right) 2^{j-1}\\
  &=O\left(k^{3/2}\log^{1/2} k+ k^2 \sum_{j=\lfloor \log A\rfloor}^{\lfloor \log k\rfloor} \frac{j}{2^j}\right)\\
  &=O\left(k^{3/2}\log^{1/2} k + k^2 \frac{\log A}{A} \right)\\
  &=O\left(k^{3/2}\log^{1/2} k + k^2 \frac{\log k}{\sqrt{k \log k}} \right)\\
  &= O(k^{3/2}\log^{1/2} k).
\end{align*}
The completes the proof of (i). It remains to prove (ii).

Recall that we set $n=\sqrt{ck\log k}$ for a sufficiently large constant $c\geq 2$, then $a(n)\leq \frac{k}{2}$. That is, at most $\frac{k}{2}$ of $k$ intervals are disjoint from $F_n$. For the remaining at least $\frac{k}{2}$ intervals, we have $q_i\leq n\leq O(\sqrt{k\log k})$, and the sum of $q_i^3$ terms for these indices $i$ is $O(k^{5/2}\log^{3/2}k)$. It remains to bound the sum of $q_i$ terms in the intervals disjoint form $F_n$. Let $A=\sqrt{ck\log k}$. Then
\begin{align*}
\sum_{i=1}^k q_i^3
  &=O(k^{5/2}\log^{3/2} k)+\sum_{j=\lfloor \log A\rfloor}^{\lfloor \log k\rfloor} \big(a(2^j) - a(2^{j-1})\big) 2^{3j}\\
  &=O(k^{5/2}\log^{3/2} k)+\sum_{j=\lfloor \log A\rfloor}^{\lfloor \log k\rfloor} a(2^j) \big(2^{3j}-2^{3(j-1)}\big)\\
  &\leq O(k^{5/2}\log^{3/2} k)+\sum_{j=\lfloor \log A\rfloor}^{\lfloor \log k\rfloor} O\left(\frac{k^2 j}{2^{2j}}\right) 2^{3j}\\
  &=O\left(k^{5/2}\log^{3/2} k+ k^2 \sum_{j=\lfloor \log A\rfloor}^{\lfloor \log k\rfloor} j\cdot 2^{j}\right)\\
  &=O\left(k^{5/2}\log^{3/2} k + k^2 \cdot k\log k \right)\\
  &= O(k^3\log k),
\end{align*}
as claimed.
\end{proof}

\begin{lemma}\label{lem:grid}
For a positive integer $k\in \N$, consider the subdivision of the plane into $k$ cones, each with apex at the origin $o$, and aperture $2\pi/k$. For each $i=1,\ldots , k$, let $p_i\in \Z^2\setminus \{0\}$ be a point in the $i$th cone that lies closest to the origin. Then $\sum_{i=1}^k \|op_i\| =O(k^{3/2}\log^{1/2} k)$.
\end{lemma}
\begin{proof}
The function $\tan(x)$ is a monotone increasing bijection from the interval $[0,\frac{\pi}{4}]$ to $[0,1]$. Its derivative is bounded by $1\leq \tan'(x)\leq 2$. Consequently, it distorts the length of any interval by a factor of at most 2. That is, it maps any interval $[\alpha,\beta]\subset [0,\frac{\pi}{4}]$ to an interval $[\tan(\alpha), \tan(\beta)]\subset [0,1]$ with $\beta-\alpha\leq \tan(\beta)-\tan(\alpha)\leq 2(\beta-\alpha)$.

Recall that Algorithm \textsc{SparseYao} operates on $k=\Theta(\eps^{-1/2})$ cones $C_1,\ldots ,C_k$.
The two coordinate axes and the lines $y=\pm x$ subdivide the plane into eight cones with aperture $\pi/4$.
These four lines contain a set $R$ of eight rays emanating from the origin. A cone $C_i$ that contains any ray in $R$ necessarily contains a lattice point $p_i$ with $\|op_i\|\leq \sqrt{2}$. Consider the cones $C_i$ between two consecutive rays in $R$; there are $O(k)$ such cones. Rotate these cones by a multiple of $\frac{\pi}{4}$ so that they correspond to slopes in the interval $[0,1]$. Due to the rotational symmetry of $\Z^2$, the rotation is an isometry on $\Z^2$.
Each cone $C_i$ is an interval of angles $[\alpha_i,\beta_i]\subset [0,\frac{\pi}{4}]$ with $\beta_i-\alpha_i=\frac{2\pi}{k}$. As noted above, this corresponds to an interval of slopes $[\tan(\alpha_i),\tan(\beta_i)]\subset [0,1]$ with
$\frac{2\pi}{k}\leq \tan \beta_i-\tan\alpha_i\leq \frac{4\pi}{k}$.

Subdivide the interval $[0,1]$ into $k$ subintervals of length $\frac{1}{k}$.
Each interval $[\tan(\alpha_i),\tan(\beta_i)]\subset [0,1]$
contains at least one intervals of this subdivision.
Let $a_i$ be the smallest positive integer such that there exists a
rational $\frac{b_i}{a_i}$ in the $i$th interval.
Then the lattice point $(a_i,b_i)$ lies in $C_i$.
Since $\frac{b_i}{a_i}\leq 1$, then $b_i\leq a_i$,
and so the distance between $(a_i,b_i)$ and the origin is at most $\sqrt{2}\cdot a_i$.
Let $p_i$ be a point in $C_i$ closest to the origin,
hence $\|op_i\|\leq \sqrt{2}\cdot a_i$.
By Lemma~\ref{lem:num}(i), we conclude that
$\sum_{i} \|op_i\| \leq \sqrt{2} \sum_i a_i = O(k^{3/2}\log^{1/2} k)$,
where the summation is over cones bounded by two consecutive rays in $R$.
Summation over all octants readily yields
$\sum_{i=1}^k \|op_i\| \leq O(k^{3/2}\log^{1/2} k)$.
\end{proof}

\begin{corollary}\label{cor:Yao}
For positive integers $k$ and $n$, let $G_{k,n}$ be the Yao-graph with $k$ cones on the $n^2$ points
in the $n\times n$ section of the integer lattice.
Then $\|G_{k,n}\| \leq  O(n^2k^{3/2}\log^{1/2} k)$.
\end{corollary}
\begin{proof}
Recall the Yao-graph is defined as a union of $n$ stars, one for each vertex.
For each vertex $p$, a star $S_p$ is obtained by subdividing the plane into
$k$ cones with apex $p$ and aperture $2\pi/k$, and connecting $p$
to the closes point in the cone (if such a point exists).
By Lemma~\ref{lem:grid}, $\|S_p\|\leq O(k^{3/2}\log^{1/2} k)$ for any $p\in \Z^2$.
Summation over $n^2$ vertices yields $\|G_{k,n}\| \leq  O(n^2k^{3/2}\log^{1/2} k)$.
\end{proof}

For every finite point set $S\subset \R^2$, the Yao-graph $Y_k(s)$ with $k=\Theta(\eps^{-1})$ cones per vertex
is a $(1+\eps)$-spanner. Corollary~\ref{cor:Yao} readily implies that the $n\times n$ section of the integer lattice admits a $(1+\eps)$-spanner of weight $O(\eps^{-3/2}\log^{1/2}(\eps^{-1})\cdot n^2)$.

\subsection{The Weight of Sparse Yao-Graphs for the Grid}
\label{ssec:next}

\begin{theorem}\label{thm:UBgrid}
Let $S$ be the $n\times n$ section of the integer lattice for some positive integer $n$.
Then the graph $G=$\textsc{SparseYao}$(S,\eps)$ has weight $O(\eps^{-1}\log(\eps^{-1})\cdot n^2)$.
\end{theorem}
\begin{proof}
The edges $p q_i(p)$ in Algorithm \textsc{SparseYao} form a Yao-graph $Y_k(S)$ with $k=\Theta(\eps^{-1/2})$ cones per vertex. By Corollary~\ref{cor:Yao}, we have $\|Y_k(S)\|\leq O(\eps^{-3/4}\log^{1/2}(\eps^{-1})\cdot n^2)$.

Algorithm \textsc{SparseYao} refines each cone $C_i(p)$ of apex $p$ and aperture $2\pi/k$ into $k$ cones $C_{i,j}(p)$, and adds some of the edges between $p$ and a closest point $q_{i,j}(p)\in C_{i,j}(p)$.
It remains to bound the weight of the edges $pq_{i,j}$ added to the spanner. Fix $p\in S$ and $i\in \{1,\ldots , k\}$. Suppose that the algorithm adds edges from $p$ to $m$ points of the form $q_{i',j}$ for $i'\in \{i-1,i,i+1\}$ and $j\in \{1,\ldots , k\}$.
We may assume w.l.o.g.\ that $m'\geq \lfloor m/2\rfloor$ of these edges lie on the left side of the ray $\overrightarrow{pq_i}$. Label these $m'$ points in ccw order about $p$ as $t_1,\ldots ,t_{m'}$, and let $t_0=q_i$. Then the triangles $\Delta(p,t_{h-1},t_h)$, for $h=1,\ldots , m'$, are interior-disjoint, contained in $W_i\setminus W_2$, and spanned by lattice points in the $\Z^2$. By Pick's theorem, $\area(\Delta(p,t_{h-1},t_h))\geq \frac12$ for all $h=1,\ldots , m'$. Consequently, $\area(W_1\setminus W_2)\geq m'/2\geq \Omega(m)$.

We also derive an upper bound on $\area(W_1\setminus W_2)$. Recall (cf.\ Fig.~\ref{fig:wedge}) that $B(p,q_i) = W_1\cap W_2$, where $W_1$ and $W_2$ are cones centered at $p$ and $q_i$, with apertures $\frac12\, \sqrt{\eps}$ and $\sqrt{\eps}$, respectively. Note that $pq_i$ partitions $W_1\setminus W_2$ into two congruent isosceles triangles, each with two sides of weight $\|pq_i\|$, and so $\area(W_1\setminus W_2)=2\|pq_i\|^2\sin (\pi-\sqrt{\eps}/2) = O(\eps^{1/2}\,\|pq_i\|^2 )$.
By contrasting the lower and upper bounds for $\area(W_1\setminus W_2)$, we obtain
\[\Omega(m)\leq \area(W_1\setminus W_2)\leq O(\eps^{1/2}\,\|pq_i\|^2 )
\,\,\, \Rightarrow \,\,\, m\leq O(\eps^{1/2} \|pq_i\|^2).
\]

By Lemma~\ref{lem:deltoid}, if an edge $pq_{i',j}(p)$ was added in the same iteration as $pq_i(p)$, then $\|pq_{i',j}(p)\|< 2\,\|pq_i(p)\|$. Overall, the total weight of all edges added together with $pq_i$ is $O(m\, \|pq_i\|)\leq O(\eps^{1/2}\, \|pq_i\|^3)$. Summation over $i=1,\ldots ,k$ yields $O(\eps^{1/2}\, \sum_{i=1}^k \|pq_i\|^3)$.

Using Lemma~\ref{lem:num}(ii) with $k=\Theta(\eps^{-1/2})$, the total weight of the edges added for a vertex $p$ is
$O(\eps^{1/2}\cdot k^3\log k) = O(\eps^{-1}\log \eps^{-1})$. Summation over all $n^2$ vertices $p\in S$,
gives an overall weight of $\|G\|=O(\eps^{-1}\log(\eps^{-1})\cdot n^2)$, as claimed.
\end{proof}

The combination of Lemma~\ref{lem:LBgrid} and Theorem~\ref{thm:UBgrid} (lower and upper bounds) establishes Theorem~\ref{thm:grid}. We restate Theorem~\ref{thm:UBgrid} for $n$ points in the unit square $[0,1]^2$.

\begin{corollary}
Let $S$ be the $\lfloor \sqrt{n}\rfloor \times \lfloor\sqrt{n}\rfloor$ section of the scaled lattice $\frac{1}{\lfloor \sqrt{n}\rfloor}\,\Z^2$, contained in $[0,1]^2$.
Then $G=\textsc{SparseYao}(S,\eps)$ has weight $O(\eps^{-1}\log(\eps^{-1})\cdot \sqrt{n})$.
\end{corollary}
\begin{proof}
Apply Theorem~\ref{thm:UBgrid} for the $\lfloor \sqrt{n}\rfloor \times \lfloor\sqrt{n}\rfloor$ section of the integer lattice $\Z^2$, and then scale down all weights by a factor of $\frac{1}{\lfloor \sqrt{n}\rfloor}$.
\end{proof}

\section{Conclusions}
\label{sec:con}

The \textsc{SparseYao} algorithm, introduced here, combines features of Yao-graphs and greedy spanners. It remains an open problem whether the celebrated greedy algorithm~\cite{althofer1993sparse} always returns a $(1+\eps)$-spanner of weight $O(\eps^{-3/2}\sqrt{n})$ for $n$ points in the unit square (and $O(\eps^{(1-d^2)/d}n^{(d-1)/d})$ for $n$ points in $[0,1]^d$). The analysis of the greedy algorithm is known to be notoriously difficult~\cite{FiltserS20,LeS22}.
It is also an open problem whether \textsc{SparseYao} or the greedy algorithm approximates the instance-minimal weight of a Euclidean $(1+\eps)$-spanner within a ratio $o_d(\eps^{-d})$ in $\R^d$, which would go below the best possible bound implied by lightness~\cite{LeS22}.

All results in this paper pertain to the Euclidean distance (i.e., $L_2$-norm in $\R^d$). Generalizations to $L_p$-norms for $p\geq 1$ (or Minkowski norms with respect to a centrally symmetric convex body in $\R^d$) would be of interest. It is unclear whether some or all of the machinery developed here generalizes to other norms.
Finally, we note that Steiner points can substantially improve the weight of a $(1+\eps)$-spanner in Euclidean space~\cite{BhoreT22,LeS22,LeS20}. It is left for future work to study the minimum weight of a Euclidean Steiner $(1+\eps)$-spanner for $n$ points in the unit cube $[0,1]^d$; and for an $n\times \ldots \times n$ section of the integer lattice $\Z^d$.

The \textsc{SparseYao}$(S,\eps)$ algorithm is a modified version of the classical Yao-graph on a set $S$ of $n$ points in $\R^d$. It can easily be implemented in $O_d(n^2\log n)$ time by a brute-force algorithm after sorting the edges of the complete graph on $n$ points. Note, however, that Yao-graphs in the plane can be constructed in $O(n\log n)$ time~\cite{ChangHT90,Funke023}; and $\Theta$-graphs in $\mathbb{R}^d$ can be computed $O_d(n\eps^{1-d}\log^{d-1}n)$ time~\cite{NS-book}: These bounds are near-linear in $n$ for constant $\eps>0$.
It remains an open problem whether the algorithm \textsc{SparseYao}$(S,\eps)$ or its adaptation to $\Theta$-graphs can be implemented in near-linear time for constant $\eps>0$ and dimension $d\in \mathbb{N}$.
For a $\lfloor n\rfloor \times \lfloor n\rfloor$ section of the integer lattice in the plane, the \textsc{SparseYao} algorithm can be implemented in $O(n(\log n+\eps^{-1}))$ time due to symmetry: One can find a closest lattice points to the origin in each cone in $O(n\log n)$ time, and then translate the resulting star to all other lattice points.

\bibliographystyle{plainurl}
\bibliography{a-revised}

\end{document}